\documentclass[12pt]{article}

\usepackage{amssymb}
\usepackage{latexsym, amsmath, amscd,amsthm,booktabs,array,url}
\usepackage[all]{xy}

\newtheorem{theorem}{Theorem}[section]
\newtheorem{corollary}[theorem]{Corollary}
\newtheorem{proposition}[theorem]{Proposition}
\newtheorem{lemma}[theorem]{Lemma}

\theoremstyle{definition}
\newtheorem{definition}[theorem]{Definition}
\newtheorem{example}[theorem]{Example}
\newtheorem{remark}[theorem]{Remark}

\newcommand{\ZZ}{{\mathbb Z}}
\newcommand{\ff}{{\mathbb F}}

\newcommand{\bA}{{\bar A}}

\newcommand{\bC}{{C'}}

\def\Q{\mathcal{Q}}

\def\wt{\mathop{\operatorfont wt}\nolimits}

\def\Res{\mathop{\operatorfont Res}\nolimits}

\newcommand{\Cc}{{\mathcal C}}

\def\Res{\mathop{\operatorfont res}\nolimits}
\def\Pic{\mathop{\operatorfont Pic}\nolimits}

\def\supp{\mathop{\operatorfont supp}\nolimits}

\begin{document}

\title{Distance bounds for algebraic geometric codes}

\author{Iwan Duursma\footnotemark[1],\, Radoslav Kirov\footnotemark[1],\,  
and Seungkook Park\footnotemark[2]} 


\date{\today}

\renewcommand{\thefootnote}{\fnsymbol{footnote}}
\footnotetext[1]{Dept. of Mathematics, University of Illinois
at Urbana-Champaign (\{duursma,rkirov2\}@math.uiuc.edu)}
\footnotetext[2]{School of Computational Sciences, Korea Institute
for Advanced Study (napsk71@kias.re.kr)}
\renewcommand{\thefootnote}{\arabic{footnote}}

\maketitle

\begin{abstract}
Various methods have been used to obtain improvements of the Goppa lower bound for the
minimum distance of an algebraic geometric code. The main methods divide into two categories
and all but a few of the known bounds are special cases of either the Lundell-McCullough floor 
bound or the Beelen order bound. The exceptions are recent improvements of the floor bound
by G{\"u}neri-Stichtenoth-Taskin, and Duursma-Park, and of the order bound by
Duursma-Park and Duursma-Kirov. In this paper we provide short proofs for all floor bounds 
and most order bounds in the setting of the van Lint and Wilson AB method. Moreover, 
we formulate unifying theorems for order bounds and formulate the DP and DK order bounds
as natural but different generalizations of the Feng-Rao bound for one-point codes.
\end{abstract}

\section*{Introduction}
 
Various methods have been used to obtain improvements of the Goppa lower bound for the
minimum distance of an algebraic geometric code. 
The best known lower bounds appear in the diagram below. Apart from the basic bounds, 
they divide into floor bounds, order bounds, and bounds of mixed type.
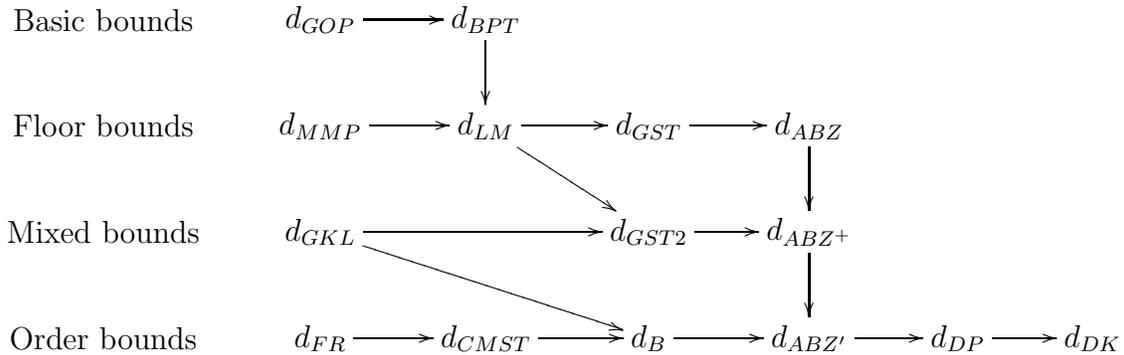
\[
\xymatrix{
\text{Basic bounds} &d_{GOP} \ar[r] &d_{BPT} \ar[d] \\
\text{Floor bounds} &d_{MMP} \ar[r] &d_{LM} \ar[r] \ar[dr] &d_{GST} \ar[r] &d_{ABZ} \ar[d] \\
\text{Mixed bounds} &d_{GKL}\ar[rr] \ar[drr] & &d_{GST2} \ar[r] &d_{ABZ^+} \ar[d] \\
\text{Order bounds} &d_{FR} \ar[r] &d_{CMST} \ar[r] &d_{B} \ar[r] &d_{ABZ'} \ar[r] &d_{DP} \ar[r] &d_{DK} \\
}
\] 
In the first part of the paper, we recall the AB method and show how it improves as well as unifies bounds. Without the
AB method, the best bounds in each category - the floor bound $d_{GST}$ \cite{GST09}, the mixed bound $d_{GST2}$ \cite{GST09}, 
and the order bound $d_B$ \cite{Bee07} - are
not comparable. The codes in Table \ref{Table:notcomp}, constructed with the Suzuki curve over $\ff_8$, illustrate that the bounds
$d_{GST}$, $d_{GST2}$ and $d_B$, can not be compared in general. 
Their improvements $d_{ABZ} \geq d_{GST}$ (Section~\ref{S:ABZ}), $d_{ABZ^+} \geq d_{GST2}$ (Section~\ref{S:ABZplus}), and $d_{ABZ'} \geq d_{B}$
(Section~\ref{S:ABZcoset}) satisfy $d_{ABZ'} \geq d_{ABZ^+} \geq d_{ABZ}$, for any given code. 
Thus, for the improved bounds, bounds of order type 
improve bounds of mixed type which in turn improve bounds of floor type. 
\begin{table}[h]
\begin{center}
\begin{tabular}{l@{\hspace{18mm}}r@{\hspace{3mm}}r@{\hspace{9mm}}r@{\hspace{12mm}}r@{\hspace{3mm}}r@{\hspace{3mm}}rrr}
\toprule
\multicolumn{1}{l}{Code}  &{$d_{GST}$}   &{$d_{GST2}$}   &{$d_{B}$}  &{$d_{ABZ}$}   &{$d_{ABZ+}$}   &{$d_{ABZ'}$} \\
\midrule
$C_\Omega(D,G=28P+2Q)$    &8   &8  &7  &8  &8  &8  \\
$C_\Omega(D,G=30P)$       &7   &6  &8  &7  &7  &8  \\
$C_\Omega(D,G=30P+Q)$     &7   &8  &8  &8  &8  &8  \\
\midrule
$C_\Omega(D,G=30P+2Q)$    &9   &9  &9  &10  &10  &10  \\
\bottomrule 
\end{tabular}
\end{center}
\caption{Suzuki curve over $\ff_8$}
\label{Table:notcomp}
\end{table} 

The best bounds overall are the order bounds $d_{DP}$ \cite{DuuPar09} and $d_{DK}$ \cite{DuuKir09}. 
In the second part of the paper we present
a framework to derive bounds of order type including the bounds $d_{DP}$ and $d_{DK}$. In Section \ref{S:semigps} and Section \ref{S:mainthm} we outline our
approach and we develope our main tools (Proposition \ref{P:rec}, Theorem \ref{T:S'}, and Theorem \ref{T:MTS}). 
Theorem \ref{T:DK} in Section \ref{S:order} gives a general order bound that includes the bound $d_{DK}$. The bounds
$d_{DP}$ and $d_B$ follow as special cases but in a form that is different from their original formulation. 
In Section \ref{S:orderpath} we show that the different formulations are equivalent. In Section \ref{S:comp} we indicate how bounds in this paper
can be computed efficiently. In the remainder of this introduction, we briefly discuss each of the three types of bounds. \\

(Floor bounds) For a divisor $H$ with $L(H) \neq 0$, its floor is the unique divisor $\lfloor H \rfloor$ that is minimal
with the property $L(H) = L(\lfloor H \rfloor)$ \cite{MahMat06}. The 
difference $E_H = H - \lfloor H \rfloor$ is called the fixed part of the divisor $H$ \cite{Mir95}. 
Maharaj, Matthews and Pirsic \cite{MMP05} showed that, for a geometric Goppa code 
$C_\Omega(D,H+\lfloor H \rfloor)$, the actual minimum distance exceeds the Goppa minimum distance
by at least the degree of the fixed part $E_H$ of $H$ (the bound $d_{MMP}$). This generalizes results
in  \cite{CarTor05}, \cite{HomKim01}. Lundell and McCullough \cite{LunMcc06} gave a further
genralization (the bound $d_{LM}$) that includes as special cases other bounds in \cite{CarTor05}, 
\cite{HomKim01}, as well as bounds in \cite{GarLax92}, \cite{KirPel95}. Recently, 
G{\"u}neri, Stichtenoth, and Taskin \cite{GST09}, and Duursma and Park \cite{DuuPar09}
gave further improvements $d_{GST}$ and $d_{ABZ}$, respectively. 
The $d_{GST}$ bound further exploits the floor bound method. The $d_{ABZ}$ bound uses an argument similar
to the AB method of van Lint and Wilson \cite{LinWil86}. In Section \ref{S:ABZ}, 
we compare the improvements and show that $d_{ABZ} \geq d_{GST} \geq d_{LM}.$ \\

(Order bounds) The Feng-Rao decoding algorithm for one-point codes corrects errors 
up to half the Goppa designed minimum distance \cite{FenRao93}, \cite{Duu93}. Soon after the algorithm was
presented it became clear that in many cases it corrects beyond half the Goppa designed minimum distance.
An analysis of the actual performance of the algorithm led Kirfel and Pellikaan to
define the Feng-Rao bound $d_{FR}$ for the minimum distance of one-point codes \cite{KirPel95}. For Hermitian one-point codes,
the bound agrees with the actual minimum distance of the code \cite{YanKum92}, \cite{KirPel95}. Later, 
the bound was connected to order domains and became known as the order bound \cite{HoeLinPel98}. 
The formulation of the order bound for general codes from curves (the bound $d_B$) is due to Beelen \cite{Bee07}. 
The bound $d_B$ agrees, for all Hermitian two-point codes, with the actual minimum distance of the code
\cite{HomKim06}, \cite{Bee07}, \cite{Par09}. Using an approach similar to that in \cite{HoeLinPel98}, 
Carvalho, Munuera, da Silva, and Torres \cite{CMST07} formulated an order bound $d_{CMST}$ for multi-point codes.
All order bounds for a code use a filtration of subcodes of the code. 
For the Feng-Rao bound the
filtration is determined by the choice of a point $P$ and takes the form
\[
C_\Omega(D,G) \supset C_\Omega(D,G+P) \supset C_\Omega(D,G+2P) \supset \cdots \supset \{0\}
\]
The bounds in  \cite{HoeLinPel98}, \cite{CMST07} follow this choice. Beelen allows the addition of
different points at different steps in the filtration. This is essential in order to attain the 
actual minimum distance of Hermitian two-point codes and in general greatly improves the order bound.
The improved bounds $d_{ABZ'}, d_{DP}$ \cite{DuuPar09} and $d_{DK}$ \cite{DuuKir09} satisfy
$d_{DK} \geq d_{DP} \geq d_{ABZ'} \geq d_B$ and $d_{ABZ'} \geq d_{ABZ}$. The bound $d_{ABZ'}$ provides a 
connection between the families of floor bounds and order bounds. It shows that in general order bounds 
provide better bounds than floor bounds. With hindsight, the bounds $d_{DP}$ and $d_{DK}$ are each 
natural generalizations of the Feng-Rao bound. The bound $d_{DP}$ generalizes the performance aspect of the bound. 
Decoding up to half the bound $d_{DP}$ is possible in much the same way as the 
original Feng-Rao decoding algorithm \cite{DuuPar09}. The bound $d_{DK}$ generalizes the bound itself, but in a way
that is no longer compatible with the original decoding algorithm. And decoding up to half the bound 
$d_{DK}$ is an open problem. \\

(Mixed bounds) The Garcia-Kim-Lax bound $d_{GKL}$ \cite{GKL93} resembles floor bounds but in some cases improves on them. 
The bound uses extra assumptions and the original proof has some characteristics of the order bound. In particular, the proof
deals separately with words in $C_\Omega(D,G) \backslash C_\Omega(D,G+P)$ as in the first step of the filtration that
is used in the order bound. G{\"u}neri, Stichtenoth and Taskin \cite{GST09} give a generalization $d_{GST2}$ of the bound $d_{GKL}$ 
that includes and improves the bound $d_{LM}$. 
We give a further improvement $d_{ABZ^+} \geq d_{GST2}$ that shows the role 
of mixed bounds as an intermediate between floor bounds and order bounds. In particular, $
d_{ABZ'} \geq d_{ABZ^+} \geq d_{ABZ}.$ Improvements of mixed bounds over similar floor bounds are in general
small. The improvement of $d_{GST2}$ over $d_{LM}$ is at most one and the improvement of $d_{ABZ^+}$ over 
$d_{ABZ}$ is at most two. We also show that the bound $d_{GKL}$ can be obtained as a special case of the bound $d_B$.

\section{Algebraic geometric codes} \label{S:AGcodes}

The following notation will be used. Let $X/\ff$ be an algebraic curve (absolutely irreducible, smooth,
projective) of genus $g$ over a finite field $\ff$. Let $\ff(X)$
be the function field of $X/\ff$ and let $\Omega(X)$ be the
module of rational differentials of $X/\ff$. 
Given a divisor $E$ on $X$ defined over $\ff$, let $L(E) = \{ f \in \ff(X) \backslash \{0\} :
(f)+E\geq 0 \} \cup \{0\},$ and let $\Omega(E) = \{ \omega \in \Omega(X) \backslash\{0\} : 
(\omega) \geq E \} \cup \{0\}.$ Let $K$ represent
the canonical divisor class. 
For $n$ distinct rational points $P_1, \ldots, P_n$ on $X$ and for disjoint divisors
$D=P_1+\cdots+P_n$ and $G$, the geometric Goppa codes $C_L(D,G)$ and $C_\Omega(D,G)$
are defined as the images of the maps
\begin{align*}
\alpha _L ~:~ &L(G)~~\longrightarrow~~\ff^{\,n}, ~~f \mapsto (\,f(P_1), \ldots, f(P_n) \,), \\
\alpha _\Omega ~:~ &\Omega(G-D)~~\longrightarrow~~\ff^{\,n}, 
~~\omega \mapsto (\,\Res_{P_1}(\omega), \ldots, \Res_{P_n}(\omega) \,).
\end{align*}
The condition that $G$ has support disjoint from $D$ is not essential and can be
removed by modifying the encoding maps $\alpha_L$ and $\alpha_\Omega$ locally at the
coordinates $P \in \supp G \cap \supp D$ \cite{TsfVla07}. 
The {Hamming distance} between two vectors $x, y \in \ff^n$ is $d(x,y) = | \{ i : x_i \neq y_i \} |.$
The {minimum distance} of a nontrivial linear code $\Cc$ is
$d(\Cc) = \min\, \{ d(x,y) : x,y \in \Cc, x \neq y \}.$ 

\begin{proposition} (Goppa bound $d_{GOP}$). 
\begin{align*}
&d(C_L(D,G)) \geq \deg\,(D-G), \quad \text{and} \quad \\ 
&d(C_\Omega(D,G)) \geq \deg\,(G-K).
\end{align*} 
\end{proposition}

Every algebraic geometric code can be represented in either of the two forms but the choice
of the representation is irrelevant for our bounds. Two codes
$C_L(D,G^\ast)$ and $C_\Omega(D,G)$ are equivalent if $G+G^\ast \sim K+D$ \cite{TsfVla07}. Our bounds 
depend on the divisor class $C$, where $C = D-G$ for a code $C_L(D,G)$
and $C = G-K$ for a code $C_\Omega(D,G)$. The codes $C_\Omega(D,G)$ and $C_L(D,G^\ast)$ 
share the same divisor class $C = G-K = D-G^\ast$ and thus bounds that depend only on the
divisor class $C$ are independent of the choice of the representation of the code.
The divisor $D$, which is the same for $C_\Omega(D,G)$ and for $C_L(D,G^\ast)$, only plays a minor 
role in the bounds. For each bound there is a finite set $S$ of points such that
the bound holds whenever $D$ is disjoint form $S$. In particular, the Goppa bound becomes 
$d \geq \deg C,$ for $S = \emptyset.$ The Goppa bound is also called the designed minimum 
distance of the code and we call the divisor $C$ the designed minimum support of the code.

\begin{proposition} (Base point bound $d_{BPT}$) If the divisor $C$ has a base point $P$, i.e. $L(C) = L(C-P)$, 
then a code with designed minimum support $C$ and defined with a divisor $D$ disjoint from $P$ has 
distance $d \geq \deg C + 1.$
\end{proposition}

\begin{proof}
There exists a word in the code of weight $w = \deg C$ if and only if $C \sim P_{i_1} + \cdots + P_{i_w}$ 
for $w$ distinct points $P_{i_1}, \ldots, P_{i_w} \in \supp(D).$ The existence of such a word would
imply $L(C) \neq L(C-P)$. Therefore $d > \deg C.$
\end{proof}

The bound applies to a code $C_\Omega(D,G)$ with $G=A+B+P$ such that $L(A+P) = L(A)$ and $L(B+P) = L(B)$,
which is essentially the case considered in \cite[Theorem 2.1]{GarLax92}. 
 
\begin{lemma} For a given divisor $G$ and a point $P$, there exist divisors $A$ and $B$ such that  
$G = A+B+P$ and $L(A+P) = L(A)$ and $L(B+P) = L(B)$ if and only if $L(C) = L(C-P),$ for $G \sim K+C.$
\end{lemma}

\begin{proof} The if part is clear, for we can choose $A=C-P$ and $B=K.$ For the only if part we use
$K-C+P \sim (K-A) + (K-B)$. And since $L(K-A) \neq L(K-A-P)$ and $L(K-B) \neq L(K-B-P)$,
$L(K-C+P) \neq L(K-C)$, or $L(C) = L(C-P).$ 
\end{proof}

\section{Floor bounds} \label{S:ABZ}

We present the ABZ floor bound of Duursma and Park \cite{DuuPar09} and show that it includes the bounds
$d_{LM}$ and $d_{GST}.$ The following lemma contains the main idea. 

\begin{lemma} \label{L:ABZ}
Given a divisor $G$, let $\eta$ be a nonzero differential with divisor 
$(\eta)=G-D'+E$, such that $D', E \geq 0$ and $E \cap D' = \emptyset.$
For divisors $A, B,$ and $Z$, such that $G=A+B+Z$, and such that $Z \geq 0$ and $Z \cap D' = \emptyset,$ 
\[
\deg D' \geq l(A)-l(A-D') + l(B)-l(B-D').
\]
\end{lemma}

\begin{proof}
With $E, Z \geq 0$ and $E \cap D' = Z \cap D' = \emptyset$, the natural maps
\begin{align*}
L(A) / L(A-D') &\longrightarrow L(A+E) / L(A+E-D'), \\
L(B) / L(B-D') &\longrightarrow L(B+Z) / L(B+Z-D'),  
\end{align*}
are well defined and injective. Therefore
\begin{align*}
\deg D' &~=~ l(A+E) - l(A+E-D') + i(A+E-D') - i(A+E) \\
        &~=~ l(A+E) - l(A+E-D') + l(B+Z) - l(B+Z-D') \\
        &~\geq~ l(A)-l(A-D') + l(B)-l(B-D').
\end{align*}
\end{proof}

\begin{remark}
The condition $Z \geq 0$ can be replaced with the weaker condition 
$L(B) \subseteq L(B+Z)$, which does not affect the proof. However,
the weaker condition does not produce better lower bounds.
Namely, suppose that $G=A+B+Z$ is a decomposition such that $L(B) \subseteq L(B+Z)$ 
and $Z \cap D' = \emptyset.$ Let $Z=Z^+ - Z^-,$ with $Z^+, Z^- \geq 0, Z^+ \cap Z^- = \emptyset.$
Then $L(B) = L(B) \cap L(B+Z) = L(B-Z^-).$ The decomposition $G=A+(B-Z^-)+Z^+$ meets
the condition $Z^+ \geq 0$ and $Z^+ \cap D' = \emptyset$ and gives the same
lower bound,
\begin{align*}
&l(A)-l(A-D') + l(B-Z^-)-l(B-Z^--D') \\
=~ & l(A)-l(A-D') + l(B) - l(B-D').
\end{align*}
\end{remark}

When written out in terms of linear algebra, i.e. after removing the connection
to curves, the bound is essentially an application of the AB bound for linear codes \cite{LinWil86}. 
We briefly formulate the connection. For two vectors $a, b$ in $\ff^n$, 
let $a \ast b = (a_1 b_1, \ldots, a_n b_n)$ denote the Hadamard or coordinate-wise 
product of the two vectors. 

\begin{lemma} 
Let ${\cal A}$, ${\cal B}$, ${\cal C}$ $\subseteq \ff^{\,n}$ be $\ff$-linear codes of 
length $n$ such that ${\cal A} \ast {\cal B} \perp {\cal C}$, i.e. such that $a \ast b \perp c$, for all
$a \in {\cal A}, b \in {\cal B}, c \in {\cal C}.$ Then, for all $c \in {\cal C}$,  
\[
\wt(c) := \dim\, ( c \ast \ff^n) ~\geq~ \dim\, (c \ast {\cal A}) + \dim\, (c \ast {\cal B}).
\]
\end{lemma}

For $G = A+B+Z,$ and $D = P_1 + P_2 + \cdots + P_n$, such that $Z \geq 0$ and $Z \cap D = \emptyset$,
let $c \in {\cal C} = C_\Omega(D,G)$ have support in $D' \leq D$. For ${\cal A} = C_L(D,A)$ and ${\cal B} = C_L(D,B)$,
${\cal A} \ast {\cal B} \perp {\cal C}$. With $c \ast {\cal A} \simeq C_L(D',A)$ and  $c \ast {\cal B} \simeq C_L(D',B)$,
\[
\deg D' ~\geq~ \dim\, {c \ast \cal A} + \dim\, {c \ast \cal B} = l(A)-l(A-D') + l(B)-l(B-D').
\]
The definition of the codes ${\cal A}$, ${\cal B}$, and ${\cal C}$ does not require that the divisors $A, B$ and $G$ 
are disjoint from $D$, if we modify the encoding map $\alpha_L$. In that case the
inclusion $C_L(D,A) \ast C_L(D,B) \subseteq C_L(D,G)$ remains valid for the modified
codes with the assumption $D \cap Z = \emptyset.$ \\

For $G=K+C$, Lemma \ref{L:ABZ} gives a lower bound for $\deg D'$ that depends only on $C$
and the choice of the divisors $A$ and $B$ in $G=A+B+Z.$

\begin{theorem}(ABZ bound \cite[Theorem 2.4]{DuuPar09}) \label{T:ABZ}
Let $G = K+C = A+B+Z,$ for $Z \geq 0$. For $D$ with $D \cap Z = \emptyset$, 
\[
d( C_\Omega(D,G) ) ~\geq~ l(A)-l(A-C) + l(B)-l(B-C).  
\]
\end{theorem}

\begin{proof}
A word $c \in C_\Omega(D,G)$ has support $D'$ only if there
exists a nonzero differential $\eta \in \Omega(G-D') \simeq L(D'-C).$
With Lemma \ref{L:ABZ},
\begin{align*}
\deg D' &~\geq~ l(A)-l(A-D') + l(B)-l(B-D') \\
        &~\geq~ l(A)-l(A-C) + l(B)-l(B-C). 
\end{align*}
\end{proof}

Replacing $A$ with $\lfloor A \rfloor$ and $B$ with $\lfloor B \rfloor$ can only improve the lower bound for $\deg D'.$ 
And in general the bound is optimal for choices of $A$ and $B$ such that $A = \lfloor A \rfloor$ and $B = \lfloor B \rfloor.$
However, it can be useful to apply the bound with $A \neq \lfloor A \rfloor$ or $B \neq \lfloor B \rfloor$ if such a choice 
reduces the support of the divisor $Z$. The choice may then give the same bound with a less restrictive condition $D \cap Z = \emptyset.$

\begin{table}[th]
\begin{center}
\begin{tabular}{l@{\hspace{6mm}}r@{\hspace{6mm}}r@{\hspace{6mm}}r@{\hspace{6mm}}r@{\hspace{6mm}}l@{\hspace{6mm}}l}
\toprule
$G$    &$A$ &$B$   &$Z$   &$d_{ABZ}$  &\text{Condition for $D$}\\
\midrule
22P+6Q &14P &8P &6Q      &6  &$Q \not \in \supp D$ \\
22P+6Q &13P &8P &P+6Q    &6  &$P,Q \not \in \supp D$ \\
\bottomrule
\end{tabular}
\end{center}
\caption{Suzuki curve over $\mathbb{F}_8~(\lfloor 14 P \rfloor = 13P)$}
\label{Table:PQD}
\end{table}

We give two other forms for the lower bound in the theorem. Equation (\ref{eq:2}) shows that the lower bound reduces to the Goppa 
designed minimum distance $\deg C$ whenever $Z=0$. Equation (\ref{eq:3}) shows that the lower bound 
never exceeds $\deg C + \deg Z.$
\begin{align}
&l(A)-l(A-C) + l(B)-l(B-C)  \label{eq:1} \\[1ex] 
~=~ &\deg C + i(A)-i(A-C) + l(B)-l(B-C) \nonumber \\ 
~=~ &\deg C + l(B+Z-C)- l(B+Z) + l(B)-l(B-C) \label{eq:2} \\[1ex] 
~=~ &\deg C + \deg Z + i(B+Z-C)- l(B+Z) + l(B)-i(B-C) \nonumber \\
~=~ &\deg C + \deg Z + l(A)-l(A+Z) +l(B)-l(B+Z). \label{eq:3}
\end{align}

With added assumptions for the divisors $A$ and $B$ we obtain as special cases of the theorem the bounds $d_{LM}$ and $d_{GST}$.

\begin{corollary} (the bound $d_{LM}$ \cite[Theorem 3]{LunMcc06}) \label{C:LM}
Let $G = K+C = A+B+Z,$ for $Z \geq 0$, such that $L(A+Z)=L(A)$ and $L(B+Z)=L(B)$. 
For $D$ with $D \cap Z = \emptyset$, 
\[
d(C_\Omega(D,G)) \geq \deg C + \deg Z.
\]
\end{corollary}

\begin{proof}
Use Equation (\ref{eq:3}) with $L(A+Z)=L(A)$ and $L(B+Z)=L(B)$.
\end{proof}

The original floor bound by Maharaj, Matthews and Pirsic \cite{MMP05} corresponds to $A+Z=B+Z=H$ and $A=B=\lfloor H \rfloor$. 


\begin{corollary} (the bound $d_{GST}$ \cite[Theorem 2.4]{GST09}) \label{C:gst}
Let $F$ be an algebraic function field of genus $g$ with full constant
field ${\mathbb F}_q.$ Let $D = P_1 + \cdots + P_n$, where the $P_i$'s are distinct rational places
of the function field $F/{\mathbb F}_q$ and suppose that $\bA, B, \bC, Z \in \operatorname{Div}(F)$ satisfy
the following conditions:
\begin{enumerate}
\item $(\supp(\bA) \cup \supp(B) \cup \supp(\bC) \cup \supp(Z)) \bigcap \supp(D) = \emptyset,$
\item $L(\bA) = L(\bA-Z)$ and $L(B)=L(B+Z)$,
\item $L(\bC)=L(B).$
\end{enumerate}
If $G = \bA+B$, then the minimum distance $d$ of the code $C_\Omega(D,G)$ satisfies
\[
d \geq \deg G - (2g-2) + \deg Z + (i(\bA)-i(G-\bC)). 
\]
\end{corollary}

\begin{proof} 
After replacing $\bC$ with $\min(\bC,B)$ if necessary, we may assume that $B=\bC+Z'$, for $Z' \geq 0.$
The bound is the special case of Theorem \ref{T:ABZ} obtained with the decomposition 
$G = A+B+Z = (\bA-Z)+\bC+(Z+Z').$ We obtain the bound in the given form using Equation (\ref{eq:3})
with $L(B+Z)=L(B)$.   
\begin{align*}
d &\geq \deg G - (2g-2) + \deg (Z+Z') + l(\bA-Z) - l(\bA+Z') \\
        &= \deg G - (2g-2) + \deg Z + \deg Z' + l(\bA) - l(\bA+Z') \\
        &= \deg G - (2g-2) + \deg Z + i(\bA) - i(\bA+Z').
\end{align*}
\end{proof}

\begin{example}
For $G=K+C=26P+(4P+Q)$, the choice $A=13P, B=16P, Z=P+Q$ gives $d_{LM}=d_{GST}=d_{GOP}+\deg Z = 7.$
The choice $A=13P, B=13P, Z=4P+Q$ gives $d_{ABZ} =8.$
In both cases, the choices are optimal.
\end{example}

 
The bound $d_{GST}$ is formulated in Corollary \ref{C:gst} as an improvement of the bound $d_{LM}$. 
For a choice of divisors $\bA$ and $B$ such that $d_{LM} = \deg C + \deg Z$, replacing $B$ with $\bC$ such that
$L(\bC) = L(B)$ gives an improvement $i(\bA) - i(G-\bC)$ of the bound $d_{LM}$. In general however, good estimates
for $d_{GST}$ do not necessarily arise from improving good estimates for $d_{LM}.$ In the example 
below, two different estimates for $d_{LM}$ are both improved by replacing $B$ with a divisor $\bC$. 
The optimal estimate $d_{GST}=6$ is the result of improving the weaker estimate $d_{LM}=4$.

\begin{table}[th]
\begin{center}
\begin{tabular}{l@{\hspace{6mm}}r@{\hspace{6mm}}r@{\hspace{6mm}}r@{\hspace{6mm}}r@{\hspace{6mm}}r@{\hspace{6mm}}rrr}
\toprule
$G$    &$\bA$ &$B$   &$Z$   &$\bC$ &$d_{LM}$  &$d_{GST}$ \\
\midrule
22P+6Q &17P+2Q &5P+4Q &P+2Q  &0    &5         &5 \\
22P+6Q &14P+2Q &8P+4Q &2Q    &8P   &4         &6 \\
\bottomrule
\end{tabular}
\end{center}
\caption{Suzuki curve over $\mathbb{F}_8$}
\label{T:LMGST}
\end{table} 

The efficient computation of bounds is discussed in Section \ref{S:comp}. To optimize the bound $d_{GST}$ we use it in the form below.
Corollary \ref{C:2} uses fewer parameters than Corollary \ref{C:gst} and gives the bound directly without comparing it to $d_{LM}$.

\begin{corollary} \label{C:2}
Let $G = K+C$, and let $B$ and $Z$ be divisors such that $L(B+Z)=L(B)$ and $Z \geq 0.$ 
For $D$ with $D \cap Z = \emptyset$, 
\[
d(C_\Omega(D,G)) \geq \deg C + l(B+Z-C)-l(B-C).
\]
\end{corollary}

\begin{proof} 
Use Equation (\ref{eq:2}) with $L(B+Z)=L(B)$.
\end{proof}

The following theorem gives the same bound as that in Corollary \ref{C:2} and Corollary \ref{C:gst} but using only a single parameter. 

\begin{theorem}(One parameter formulation of $d_{GST}$) \label{T:GST1p}
Let $G = K+C.$ For divisors $D$ and $B$ such that $D \cap (B-\lfloor B \rfloor) = \emptyset$,  
\[
d(C_\Omega(D,G)) \geq \deg C + l(B-C)-l( \lfloor B \rfloor -C).
\]
\end{theorem}

\begin{proof}
Let $B = \lfloor B \rfloor + Z$, $Z \geq 0$. The theorem follows by applying Corollary \ref{C:2}.
\end{proof}

The comment after Theorem \ref{T:ABZ} applies. If $B'$ is a divisor with $\lfloor B \rfloor \leq B' \leq B$
such that $L(B'-C) = L(\lfloor B \rfloor - C)$ and if $B-B'$ has smaller support than $B-\lfloor B \rfloor$
then Corollary \ref{C:2} will give the same bound as Theorem \ref{T:GST1p} but with a weaker condition for $D$.

%

\section{Mixed bounds} \label{S:ABZplus}

It is clear from the proof of Theorem \ref{T:ABZ} that the lower bound $d_{ABZ}$ can be improved if
we can show that $L(A-C) \neq L(A-D')$ or $L(B-C) \neq L(B-D').$
An interesting special case that can be explained in this way is the
bound $d_{GKL}$ by Garcia, Kim, and Lax \cite{GKL93}. In \cite{GST09}, G{\"u}neri, Stichtenoth, and Taskin present a
second bound $d_{GST2}$ that includes both the bound $d_{GKL}$ and the bound $d_{LM}$. The bound $d_{GST2}$ 
applies to codes $C_\Omega(D,G)$ and uses a decomposition $G=A+B+Z$ such that $L(A+Z)=L(A),$ $L(B+Z)=L(B)$, 
as in the bound $d_{LM}$. Moreover it is assumed that $B \leq A.$ We formulate the bound $d_{ABZ^+}$ as an 
unrestricted generalization that applies to any decomposition $G=A+B+Z.$
 
\begin{lemma} \label{L:ABZ+}
For a given divisor $C$, let $P$ be a point with $L(C)=L(C-P),$ 
and let $A' \leq A$ be a pair of divisors such that
\begin{enumerate}
\item $L(A'-C) \neq L(A'-C-P)$ and $L(A')=L(A'-P),$ and 
\item $L(A-C) \neq L(A-C-Q),$ for all $Q$ with $A' \leq A-Q \leq A.$ 
\end{enumerate}
Then $L(A-C) \neq L(A-D')$ for any divisor $D' \sim C+E$ such that $D' \cap P = \emptyset$
and $E \geq 0.$
\end{lemma}

\begin{proof}
The claim follows immediately from the second condition if $D' \geq C+Q$ for some
$Q$ with $A' \leq A-Q \leq A.$ We may therefore assume that $(D'-C) \cap (A-A') = \emptyset.$
With this assumption, the natural map
\[
L(A'-C) / L(A'-D') \longrightarrow L(A-C) / L(A-D')
\]
is well defined and injective. The first condition and $D' \cap P = \emptyset$
imply that $L(A'-D') = L(A'-P-D')$. And thus
\begin{align*}
l(A-C) - l(A-D') &~\geq~ l(A'-C) - l(A'-D') \\
                        &~=~  l(A'-P-C) - l(A'-P-D') + 1 ~>~ 0.
\end{align*}
\end{proof}

\begin{theorem}(ABZ+ bound) \label{T:ABZ+}
Let $G=K+C=A+B+Z$, for $Z \geq 0$, and let $D' \sim C+E$ be a divisor
such that $D' \cap Z = \emptyset.$ Define $\delta(A) \in \{0,1\}$ to be $1$ if there exists
a divisor $A' \leq A$ such that $\supp (A-A') \subseteq \supp (Z)$ and
\begin{enumerate}
\item $(\exists P \in Z)~$ $L(A'-C) \neq L(A'-C-P)$ and $L(A')=L(A'-P),$ and 
\item $(\forall Q \in Z)~$ $L(A-C) \neq L(A-C-Q).$ 
\end{enumerate}
Then
\[
\deg D' \geq l(A)-l(A-C) + l(B)-l(B-C) + \delta(A) + \delta(B).
\]
\end{theorem}

\begin{proof}
The proof is similar to that of Theorem \ref{T:ABZ}. With Lemma \ref{L:ABZ+}, it becomes
\begin{align*}
\deg D' &~\geq~ l(A)-l(A-D') + l(B)-l(B-D'), \\
        &~\geq~ l(A)-l(A-C)+\delta(A) + l(B)-l(B-C)+\delta(B).
\end{align*}
\end{proof}

\begin{corollary} (the bound $d_{GST2}$ \cite[Theorem 2.12]{GST09}) 
Let $F$ be an algebraic function field of genus $g$ with full constant
field ${\mathbb F}_q.$ Let $D = P_1 + \cdots + P_n$, where the $P_i$'s are distinct rational places of
the function field $F/{\mathbb F}_q$ and suppose that $\bA, B, Z \in \operatorname{Div}(F)$ satisfy the following conditions:
\begin{enumerate}
\item $(\supp(\bA) \cup \supp(B) \cup \supp(Z)) \bigcap \supp(D) = \emptyset,$
\item $\supp(\bA-B) \subseteq \supp(Z),$
\item $Z \geq 0, L(\bA) = L(\bA-Z)$ and $L(B)=L(B+Z+Q)$ for all $Q \in \supp(Z),$
\item $B+Z+P \leq \bA$ for some $P \in \supp(Z).$
\end{enumerate}
If $G = \bA+B$, then the minimum distance $d$ of the code $C_\Omega(D,G)$ satisfies
\[
d \geq \deg G - (2g-2) + \deg Z + 1.
\]
\end{corollary}

\begin{proof}
For $\bA=A+Z$, the theorem applies with $G=K+C=\bA+B=A+B+Z$ and $A'=B+P.$
We write Condition 3 in the form
\[
\begin{cases}
L(\bA) = L(\bA-P), \\
L(B+P)=L(B), \\
L(B+Z)=L(B+Z+Q).
\end{cases}
~~\Leftrightarrow~~
\begin{cases}
L(A'-C) \neq L(A'-C-P), \\
L(A') = L(A'-P), \\
L(A-C) \neq L(A-C-Q).
\end{cases}
\]
\end{proof}

Compared with the corollary, the theorem does not require the conditions $L(A+Z) = L(A)$
and $L(B+Z) = L(B)$, and the choice of $A' \leq A$ is not restricted to the choice $A' = B+P.$
The removal of the last restriction means that the argument can be applied with choices
$A' \leq A$ and $B' \leq B$ with a potential gain of $+2$ instead of $+1$. 

\begin{example}
For $G=K+C=26P+(3P+Q)$, the choice $A=10P, B=18P, Z=P+Q$ gives $d_{LM}=d_{GST}=d_{GST_2}=d_{GOP}+\deg Z = 6.$
The choice $A=13P, B=13P, Z=3P+Q$ and $A'=B'=11P$ gives $d_{ABZ} =6,$ $d_{ABZ+}=8.$
In all cases, the choices are optimal.
\end{example}


\section{The order bounds $d_B$ and $d_{ABZ'}$} \label{S:ABZcoset}

For the minimum distance of a code $C_\Omega(D,G)$, the ABZ bound (Theorem \ref{T:ABZ}) gives
\[
d( C_\Omega(D,G) ) ~\geq~ l(A)-l(A-C) + l(B)-l(B-C),  
\]
where $G = K+C = A+B+Z,$ such that $Z \geq 0$ and $D \cap Z = \emptyset$. For a point $P$ disjoint from $D$, 
if $L(A)=L(A-P)$ and $L(A-C) \neq L(A-C-P)$ then replacing $A$ with $A-P$ (and $Z$ with $Z+P$) improves the 
lower bound by $1$. It turns out that the lower bound improves by $1$ for any divisor $A-iP$, $i \geq 0$,
with the same properties. To see this we need to go back to the proof of the ABZ bound.
The proof uses that a nonzero codeword has support $D'$ such that $D' \sim C+E,$ for $E \geq 0$, and
\begin{align*}
\deg D' &~\geq~ l(A)-l(A-D') + l(B)-l(B-D'), \\
        &~\geq~ l(A)-l(A-C) + l(B)-l(B-C). 
\end{align*}
As in the previous section, we obtain improvements for the ABZ bound from estimates for the differences 
$l(A-C) -l(A-D')$ and $l(B-C)-l(B-D')$. Let $\Delta'(A) \subset \{ A - iP : i \geq 0 \}$ be the subset
of divisors $A'=A-iP$ with the property that $L(A')=L(A'-P)$ and $L(A'-C) \neq L(A'-C-P).$ 
We claim that, for a support $D'$ with both $D'$ and $E$ disjoint from $P$,
\[
l(A-C) - l(A-D') \geq |\Delta'(A)|.
\]
For $A'$ such that $L(A')=L(A'-P)$ and for $D'$ disjoint from $P$, 
$L(A'-D') = L(A'-D'-P).$ If moreover $L(A'-C) \neq L(A'-C-P)$ then
\[
l(A'-C) - l(A'-D') = l(A'-C-P) - l(A'-D'-P) + 1.
\]
For a general divisor $A'$ and for $E$ disjoint from $P$, 
\[
l(A'-C) - l(A'-D') \geq l(A'-C-P) - l(A-D'-P).
\]
Therefore,
\begin{multline*}
l(A-C) - l(A-D') = \sum_{i \geq 0} \,[\, (\,l(A-C-iP) - l(A-D'-iP)\,) \\
 - (\, l(A-C-iP-P) - l(A-D'-iP-P)\,) \,]\, \geq |\Delta'(A)|.
\end{multline*}
We give a first formulation of the $ABZ'$ bound. 

\begin{theorem}($ABZ'$ bound) \label{T:ABZ'}
Let $d_{ABZ}$ be the $ABZ$ bound for $d(C_\Omega(D,G))$ obtained with a choice of divisors $A, B$ and $Z.$ 
For a rational point $P$ disjoint from $D$,
\[
d (C_\Omega(D,G)) \geq \min \{ d_{ABZ} + |\Delta'(A)| + |\Delta'(B)|,  d (C_\Omega(D,G+P)) \}.
\] 
\end{theorem}

\begin{proof} The first argument in the minimum is a lower bound when $E$ is disjoint from $P$, and the second argument 
is a lower bound when $E$ is not disjoint from $P$. 
\end{proof}

We will give a different formulation in Section \ref{S:mainthm}. An advantage of this formulation is the easy comparison with
the $ABZ$ bound for the same choice of $A, B$ and $Z$. On the other hand, the best results for the $ABZ$ bound and the $ABZ'$ 
bound are in general obtained with different choices for $A, B$ and $Z$. The formulation in Section \ref{S:mainthm} will be
easier to compare with other order bounds and easier to optimize. \\



The special case $Z=0$ of the order bound $d_{ABZ'}$ returns the Beelen bound $d_B$ (\cite{DuuPar09}, or Corollary \ref{C:ABZ'2}). 
The special case $Z=0$ of the floor bound $d_{ABZ}$ returns the Goppa bound $d_{GOP}$.
The bounds $d_{GOP},d_B$ are therefore in the same relation as the bounds $d_{ABZ}, d_{ABZ'}$ and follow
from the latter as the special case $Z=0$. 
\[
Z=0:~~d_{GOP} \longrightarrow d_B \qquad \qquad Z \geq 0:~~d_{ABZ} \longrightarrow d_{ABZ'} 
\] 

\begin{example}
The bounds in Table \ref{T:A'B'} all use a choice $A= B = 13 P$ (so that $Z=2P+2Q, P+2Q, P+Q,$ respectively).
In all cases this is an optimal choice. The gains for $d_{ABZ^+}, d_{ABZ'}$ in the second row use $A', B' \in \{ 11P \}$.
The gains for $d_{ABZ^+}, d_{ABZ'}$ in the last row use $A', B' \in \{9P, 11P\}$. In partciular,
$d_{ABZ'}=8$ uses $d_{ABZ'} = \min \{ 4+2+2, 8 \} = 8.$
\end{example}

\begin{table}[th]
\begin{center}
\begin{tabular}{l@{\hspace{18mm}}r@{\hspace{3mm}}r@{\hspace{3mm}}r@{\hspace{9mm}}r@{\hspace{3mm}}r@{\hspace{9mm}}rrr}
\toprule
\multicolumn{1}{l}{Code}  &$d_{LM}$ &{$d_{GST}$}  &$d_{ABZ}$  &{$d_{GST2}$}   &{$d_{ABZ^+}$}   &{$d_{ABZ'}$}   \\
\midrule
$C_\Omega(D,G=28P+2Q)$    &8   &8  &8  &8  &8  &8  \\
$C_\Omega(D,G=27P+2Q)$       &6   &6  &6  &6  &8  &8  \\
$C_\Omega(D,G=27P+Q)$     &4   &4  &4  &4  &6  &8  \\
\bottomrule 
\end{tabular}
\end{center}
\caption{Suzuki curve over $\mathbb{F}_8$}
\label{T:A'B'}
\end{table} 

The bound $d_{GKL}$ is stated in terms of $H$-Weierstrass gaps at a point $P$. It is a special case
of the bound $d_{GST2}$ \cite[Corollary 2.13]{GST09}. We formulate the bound and give two different proofs, 
showing that it is also a special case of the bound $d_B$. 

\begin{theorem} (The bound $d_{GKL}$ \cite{GKL93})
Let $H$ be a divisor and let $P$ be a rational point such that, for integers $\alpha, \beta, t$ with $\beta \geq \alpha + t$ and $t \geq 1$,
\[
L(H + \alpha P + tP) = L(H + \alpha P - P), \quad L(H + \beta P)  = L(H + \beta P - tP).
\] 
Then, for $G = 2H + (\alpha+\beta-1)P$ and for $D$ disjoint from $H$ and $P$, $d(C_\Omega(D,G) \geq \deg G - (2g-2) + t+ 1.$ 
\end{theorem}

(The reduction $d_{GKL} \rightarrow d_{GST2}$) We apply the $ABZ+$ bound (Theorem \ref{T:ABZ+}). 
With $G = A + B + Z = (H + \alpha P - P) + (H+ \beta P - tP) + tP$ and $B' = H + \alpha P \leq B$,
we find $d(C_\Omega(D,G)) \geq \deg G - (2g-2) + t+1.$ \\

(The reduction $d_{GKL} \rightarrow d_{B}$) We apply the $ABZ'$ bound (Theorem \ref{T:ABZ'}) with $Z=0$. 
For $i = 0, \ldots, t,$ let $G+iP = A+B+Z = (H+ \alpha P + iP - P) + (H+\beta P) + 0.$ Then 
\[
d(C_\Omega(D,G+iP)) \geq \min \,\{\, \deg G - (2g-2) + i + |\Delta'(A)| + |\Delta'(B)|, d(C_\Omega(G+iP+P)) \,\}.
\]
With $H + \alpha P + iP, \ldots, H+ \alpha P + (t-1)P, H + \beta P \in \Delta'(B)$, we obtain $|\Delta'(B)| \geq t-i+1$, and thus
\[
d(C_\Omega(D,G)) \geq \min \,\{\, \deg G - (2g-2) + t+1, d(C_\Omega(D,G+tP+P)) \,\} \geq \deg G - (2g-2) + t+1.
\]
   
\section{Base point free semigroups} \label{S:semigps}

We will discuss in Section \ref{S:order} the various order bounds. 
First we introduce, for divisors $C$ and for sets of points $S$ and $S'$, subsets of divisor classes $\Gamma(C;S,S')$.
The sets capture the desired coding theory parameters in the language of divisors. Together with the results in the next section they
allow us to present all order bounds in a unified framework. \\

Let $X/\ff$ be a curve over a field $\ff$ and let $\Pic(X)$ be the
group of divisor classes. Let $\Gamma = \{ A : L(A) \neq 0 \}$ be
the semigroup of effective divisor classes. For a given rational
point $P \in X$, let $\Gamma_P = \{ A : L(A) \neq L(A-P) \}$ be
the semigroup of effective divisor classes with no base point at
$P$. For a finite set of points $S$, let $\Gamma_S = \cap_{P \in S} \Gamma_P$. 
By convention, let $\Gamma_\emptyset = \Gamma$.


\begin{definition}
For a divisor class $C$ and for finite sets of rational points $S$ and $S'$, let
\begin{align*}
\Gamma(C;S,S') ~=~ &\{ A : A \in \Gamma_S \text{ and } A-C \in \Gamma_{S'} \}, \\
\gamma(C;S,S') ~=~ &\min \{ \deg A : A \in \Gamma(C;S,S') \}.
\end{align*}
\end{definition}

From the definition it is clear that $\Gamma(C;S,S')$ lives inside the semigroup $\Gamma_S$. 
Moreover, $\Gamma_{S\cup S'}$ acts on $\Gamma(C;S,S')$ via divisor addition, 
and for $S' \subseteq S$, $\Gamma(C;S,S')$ is a semigroup ideal in $\Gamma_S$.
For the connection to coding theory, we have the following interpretation.

\begin{lemma} \label{L:codegammas} (\cite[Lemma 4.3, Lemma 4.2]{DuuPar09}) 
For a given set of rational points $S$, and for algebraic geometric codes defined
with a divisor $D=P_1+\cdots+P_n$ disjoint from $S$,
\begin{align*}
&d(C_L(D,G)) \geq \gamma(D-G;S,\emptyset). \\ 
&d(C_\Omega(D,G)) \geq \gamma(G-K;S,\emptyset). 
\end{align*}
Moreover, for a point $P$,
\begin{align*}
\min \wt (C_L(D,G) \backslash C_L(D,G-P)) ~\geq~ &\gamma(D-G;S,P). \\
\min \wt (C_\Omega(D,G) \backslash C_\Omega(D,G+P)) ~\geq~ &\gamma(G-K;S,P).
\end{align*}
\end{lemma}

The case of a general set $S'$ follows directly from the lemma.

\begin{proposition}
For given sets of rational points $S$ and $S'$, and for algebraic geometric codes defined
with a divisor $D=P_1+\cdots+P_n$ disjoint from $S$,
\begin{align*}
\min \wt (C_L(D,G) \backslash \bigcup_{P \in S'} C_L(D,G-P)) ~\geq~ &\gamma(D-G;S,S'). \\
\min \wt (C_\Omega(D,G) \backslash \bigcup_{P \in S'} C_\Omega(D,G+P))  ~\geq~ &\gamma(G-K;S,S').
\end{align*}
Here it is agreed, for the case $S' = \emptyset$, that an empty union of vector spaces is the null space.
\end{proposition}

\begin{proof}
The case $S' = \emptyset$ is the first part of the lemma.  The case $S' \neq \emptyset$ reduces to the second part of the lemma
if we use $\cap_{P \in S'} \Gamma(C;S,P) = \Gamma(C;S,S').$
\end{proof}

The first case of Lemma \ref{L:codegammas} is particularly important for our approach to order bounds and for that reason we recall 
the proof. There exists a nonzero word in $C_L(D,G)$ with support in $A$, 
for $0 \leq A \leq D$, if and only if $L(G-D+A)/L(G-D) \neq 0.$ Since $S$ is disjoint from $D$, it is also 
disjoint from $A$. Since $A$ is effective, $L(A)$ contains the constants, but $L(A-P)$ does not, for all $P \in S$.
So that $A \in \Gamma_S$. It is clear that $L(A-(D-G)) \neq 0$ and thus $A \in \Gamma(D-G;S,\emptyset).$ 
There exists a nonzero word in $C_\Omega(D,G)$ 
with support in $A$, for $0 \leq A \leq D$, if and only if $\Omega(G-A)/\Omega(G) \neq 0$ if and only if
$L(K-G+A) / L(K-G) \neq 0.$ The rest of the proof is similar to the previous case with $D-G$ replaced by $G-K$. \\

The bounds in Lemma \ref{L:codegammas} can be used for codes with $L(-C) = L(G-D) = 0$ or $L(-C) = L(K-G) = 0$. 
This includes all codes with a positive designed minimum distance. For codes with $L(-C) \neq 0$, we 
see that $0 \in \Gamma(C;S,\emptyset)$ and $\gamma(C;S,\emptyset)=0$. In order to obtain nontrivial lower bounds 
for such codes the set $\Gamma(C;S,\emptyset)$ should be replaced with the subset
\[
\Gamma^\ast(C;S,\emptyset) = \{ A \in \Gamma_S : L(A-C) \neq L(-C) \}, 
\]
and the lower bound $\gamma(C;S,\emptyset)$ for the minimum distance with $\gamma^\ast(C;S,\emptyset)$, where the latter
denotes the minimal degree for a divisor $A \in \Gamma^\ast(C;S,\emptyset).$ Details can be found in \cite[Section 4]{DuuPar09}. 
Proposition \ref{P:rec} and Theorem \ref{T:S'} play a key role in the definition of the order bounds in
Section \ref{S:order}.


\begin{proposition} \label{P:rec}
For $P \not \in S'$,
\[
\Gamma(C;S,S') = \Gamma(C;S,S' \cup \{P\}) \cup \Gamma(C+P;S,S').
\]
\end{proposition}

\begin{proof}
($\subseteq$) Let $D \in \Gamma(C;S,S')$. For $P \not \in S',$
\begin{align*}
&L(D-C) \neq L(D-C-P) ~\Rightarrow~ D \in \Gamma(C;S,S' \cup \{P\}). \\
&L(D-C) = L(D-C-P) ~\Rightarrow~ D \in \Gamma(C+P;S,S').
\end{align*}
($\supseteq$) Clearly, $\Gamma(C;S,S' \cup \{P\}) \subseteq \Gamma(C;S,S')$. Let $D \in \Gamma(C+P;S,S')$. Since $P \not \in S'$, $P \in \Gamma_{S'}$. Thus, using the semigroup property, $D-C-P \in \Gamma_{S'}$ implies $D-C \in \Gamma_{S'}$, which proves $\Gamma(C+P;S,S') \subseteq \Gamma(C;S,S')$.
\end{proof}

The following theorem is proved by repeated application of the proposition. 

\begin{theorem} \label{T:S'}
For $T' \cup T = S'$,
\[
\Gamma(C;S,T') = \bigcup_{\lambda \in \Lambda} \Gamma(C+\lambda;S,S'),
\]
where $\Lambda$ is the semigroup generated by the points in $T$ (including the zero divisor).
\end{theorem}

Note that both the proposition and the theorem translate into statements about $\gamma$ if we replace $\Gamma$ with $\gamma$ and $\cup$ with $\min$. 

\section{Main theorem} \label{S:mainthm}

In this section we present a general method to obtain lower bounds for $\gamma(C;S,S')$. Combined with the properties of $\Gamma(C;S,S')$ from the previous section, the
method gives lower bounds for the minimum distance. In the next section we will derive the bounds $d_{DK}$ and $d_{DP}$ in this way.

\begin{theorem}\label{T:MTS}
Given a divisor $C$ and finite sets of rational points $S$ and $S'$, let $\{ A_0, A_1,$ $\ldots,$ $A_n \}$ be a sequence of divisors
such that $A_i = A_{i-1} + P_{i}$, $P_i$ a rational point, for $i=1,\ldots,n,$ and define subsets $\Delta, \Delta', I, I' \subset \{1,2,\ldots,n\}$
as follows.
\begin{align*}
\Delta = \{ i : A_i \in \Gamma_{P_i} \mbox{ and } A_i-C \not \in \Gamma_{P_i} \},  \quad &I = \{ i : P_i \in S \}, \\
\Delta' = \{ i : A_i \not \in \Gamma_{P_i} \mbox{ and } A_i-C \in \Gamma_{P_i} \}, \quad &I' = \{ i : P_i \in S' \}. 
\end{align*}
Then $\gamma(C;S,S') \geq |\Delta \cap I'| + |\Delta' \cap I| - |\Delta'|.$
In particular, $\gamma(C;S,S') \geq |\Delta|$ for $\Delta \subseteq I'$ and $\Delta' \subseteq I.$
\end{theorem}

\begin{proof}
For an arbitrary divisor $D \in \Gamma$,
\begin{align*}
\deg D  ~\geq~ &(l(A_n)-l(A_n-D)) \\
        ~\geq~ &(l(A_n)-l(A_n-D)) - (l(A_{0})-l(A_{0}-D)) \\
          ~=~  &(l(A_n)-l(A_0)) - (l(A_{n}-D)-l(A_{0}-D)) \\
          ~=~  &\sum_{i=1}^n (l(A_i) - l(A_{i-1})) - \sum_{i=1}^n (l(A_i-D) - l(A_{i-1}-D)) \\
          ~=~  &\sum_{i=1}^n (l(A_i) - l(A_i-P_i)) - \sum_{i=1}^n (l(A_i-D) - l(A_i-D-P_i)) \\
          ~=~  &|\{i : A_i \in \Gamma_{P_i} \mbox{ and } A_i-D \notin \Gamma_{P_i} \}| - |\{i : A_i \notin \Gamma_{P_i} \mbox{ and } A_i-D \in \Gamma_{P_i} \}|. 
\end{align*}
Let $D \in \Gamma(C;S,S')$ be of minimal degree. We show that
\begin{align*}
|\{i : A_i \in \Gamma_{P_i} \mbox{ and } A_i-D \notin \Gamma_{P_i} \}|  ~\geq~ &|\Delta \cap I'|, \\
|\{i : A_i \notin \Gamma_{P_i} \mbox{ and } A_i-D \in \Gamma_{P_i} \}|  ~\leq~ &|\Delta' \backslash I| = |\Delta'| - |\Delta' \cap I|.
\end{align*}
For $i \in I'$, $D-C \in \Gamma_{P_i}.$ Using the semigroup property of $\Gamma_{P_i},$
\begin{align*}
i \in \Delta \cap I'  ~\Rightarrow~  &A_i \in \Gamma_{P_i} \text{ and } A_i-C \notin \Gamma_{P_i} \text{ and } D-C \in \Gamma_{P_i}  \\
                      ~\Rightarrow~  &A_i \in \Gamma_{P_i} \text{ and } A_i-D \notin \Gamma_{P_i}.
\end{align*}
This proves the first inequality. For $D \in \Gamma(C;S,S')$, if $D$ and $D-C$ have a common base point $P$ then $P \not \in S \cup S'$ and
$D-P \in \Gamma(C;S,S').$ Thus, for $D$ of minimal degree, no such common base point exists and $D \notin \Gamma_P$ implies $D-C \in \Gamma_P$,
for any point $P$. We can now prove the second inequality.
\begin{align*} 
&A_i \notin \Gamma_{P_i} \text{ and } A_i-D \in \Gamma_{P_i} \\
~\Rightarrow~ &A_i \notin \Gamma_{P_i} \text{ and } A_i-D \in \Gamma_{P_i} \text{ and } D \notin \Gamma_{P_i} \\
~\Rightarrow~ &A_i \notin \Gamma_{P_i} \text{ and } A_i-D \in \Gamma_{P_i} \text{ and } D \notin \Gamma_{P_i} \text{ and } D-C \in \Gamma_{P_i} \\
~\Rightarrow~ &A_i \notin \Gamma_{P_i} \text{ and } A_i-C \in \Gamma_{P_i} \mbox{ and } D \notin \Gamma_{P_i} \\
~\Rightarrow~ &i \in \Delta' \backslash I. 
\end{align*}
\end{proof}

The order bounds $d_B, d_{ABZ'}, d_{DP}, d_{DK}$ can all be obtained from the main theorem in combination with results from 
the previous section. Using the theorem with different formats for the sequence $\{ A_i \}$ yields different bounds. The bounds 
$d_{DP}$ and $d_{DK}$ use a general format. The bound $d_B$ uses the format $A_i = B+iP,$ for a fixed $B$ and for $i \in \ZZ.$
The special case $A_i = iP,$ for $i \in \ZZ,$ is used in the Feng-Rao bound and the Carvalho-Munuera-daSilva-Torres bound. 

\begin{example}
For $C=-3P+6Q$, we apply the theorem with two different sequences.
\[
\begin{array}{llll}
A_i = iP: ~~    &\Delta = \{0,8,12,13,16,24\}, ~~ &\Delta' = \{ 17, 19, 27 \}. \\
A_i = iP+3Q: ~~  &\Delta = \{0,8,11,12,13,16,24\}, ~~ &\Delta' =\{ 7, 9, 15, 17\}. 
\end{array}
\]
The translated sequence yields an improved estimate $\gamma(C;P,P) \geq 7.$ 
\end{example}

The bound $d_{ABZ'}$ uses a sequence $\{ A_i \}$ that contains the divisors $B+iP$, for $i \leq 0$, as well as the
divisors $B+Z+iP$, for a fixed divisor $Z \geq 0$ and for $i>0.$

\begin{example}
For $C=2P+2Q$, the two choices
\[
\begin{array}{llll}
A_i = iP: ~~    &\Delta = \{0,8,10,13,16,21,29 \}, ~~ &\Delta' = \{ 14, 15, 27 \}, \\
A_i = iP+2Q: ~~  &\Delta = \{0,8,13,16,19,21,29\}, ~~ &\Delta' =\{ 2, 14, 15, \}, 
\end{array}
\]
both yield $\gamma(C;P,P) \geq 7$. This is not improved with a different choice of translated sequence. 
However, for the combined sequence 
\[
A_i = 0,\, \ldots,\, 15P,\, 15P+Q,\, 15P+2Q,\, \ldots,\, 29P+2Q,
\]
we see that the divisors
$iP,$ for $i \in \{0,8,10,13\},$ as well as the divisors $iP+2Q,$ for $i \in \{16,19,21,29\},$ contribute to 
$\Delta$ and thus $\gamma(C;\{P,Q\},P) \geq 8.$ In this case $|\Delta'|=4$, with contributions by 
$14P, 15P$ (both with $P_i=P$) and $15P+Q, 15P+2Q$ (both with $P_i=Q$).
\end{example}

 
The bound $d_{ABZ'}$ is a special case of the bound $d_{DP}$. The latter applies the theorem with $S' = \{ P \}$ but with
no restrictions on the sequence $\{ A_i \}$. The bound $d_{DK}$ applies the main theorem with no restrictions on neither $S$ and $S'$ nor on the 
sequence $\{ A_i \}$. 

\begin{example} \label{E:58}
For $C=-5P+8Q$, the two choices
\[
\begin{array}{llll}
A_i = iP-3Q: ~~    &\Delta = \{10,12,13,22,23,25 \}, ~~ &\Delta' = \{ 8, 16, 27 \}, \\
A_i = iP-2Q: ~~  &\Delta = \{10,12,13,22,23,25 \}, ~~ &\Delta' =\{ 8, 19, 27, \},
\end{array}
\]
both yield $\gamma(C;P,P) \geq 6$. An arbitrary combination of translates does not produce improvements for
$\gamma(C;\{P,Q\},P) \geq 6$. However, for the combined sequence 
\[
A_i = 10P-3Q,\, \ldots,\, 16P-3Q,\, 16P-2Q,\, \ldots, 25P-2Q,
\]
the divisor $16P-2Q$ contributes to $\Delta$ with $P_i = Q$. Together with the contributions
$iP-3Q,$ for $i \in \{10,12,13\}$ and $iP-2Q,$ for $i \in \{22, 23, 25\}$ this gives $|\Delta| = 7$ and
$\gamma(C;P,\{P,Q\}) \geq 7.$ The contributions to $\Delta'$ come from $8P, 16P, 19P, 25P$ (all with $P_i=P$) and thus
the lower bound holds with $S= \{P\}.$
\end{example}

The bounds $d_{DK} \geq d_{DP} \geq d_B$ use the main theorem with the restrictions
\[
(DK)~S,\,S' \text{ finite}, \qquad (DP)~S \text{ finite},\,S' = \{P\}, \qquad (B)~S = S' = \{P\}.
\]
The bound $d_{ABZ'}$ is a special case of the bound $d_{DP}$. Its main purpose is to connect the bounds of order type with the bounds
of floor type via the relation $d_{ABZ'} \geq d_{ABZ}$. We first show how the bound $d_{ABZ'}$ follows from the main theorem and then that 
it agrees with the earlier formulation as an improvement of the floor bound. Recall from Theorem \ref{T:ABZ'} that 
\begin{equation} \label{abz'}
d (C_\Omega(D,G)) \geq \min \{ d_{ABZ} + |\Delta'(A)| + |\Delta'(B)|,  d (C_\Omega(D,G+P)) \}
\end{equation}
Here $G=K+C=A+B+Z,$ for $Z \cap D = \emptyset$, and for $P \not \in D$.
Let $\Delta(A) \subset \{ A - iP : i \geq 0 \}$ be the subset
of divisors $A'=A-iP$ with the property that $L(A') \neq L(A'-P)$ and $L(A'-C) = L(A'-C-P).$ 

\begin{corollary} ($ABZ'$ bound \cite{DuuPar09}) \label{C:ABZ'2}
Let $G=K+C=A+B+Z,$ such that $Z\geq 0$. Then
\[ 
\gamma(C;\supp(Z),P) \geq |\Delta(A)| + |\Delta(B)|.
\]
\end{corollary}

\begin{proof} Apply the main theorem with a sequence $\{ A_i \}$ that contains the divisors $B+iP$, for $i \leq 0$, as well as the
divisors $B+Z+iP$, for $i>0.$
\end{proof}

The relation between $\Delta(A)$ and $\Delta'(A)$ is such that
$\Delta(A) = l(A) - l(A-C) + \Delta'(A)$. And thus the corollary can be stated as
\[ 
\gamma(G-K;\supp(Z),P) \geq d_{ABZ} + |\Delta'(A)| + |\Delta'(B)|.
\]
Using Lemma \ref{L:codegammas} we recover the ABZ' bound in the form (\ref{abz'}). \\


It is clear from the definitions that $A \in \Gamma(C;S,S')$ if and only if $A-C \in \Gamma(-C;S',S)$, and thus 
$\gamma(C;S,S') - \gamma(-C;S',S) = \deg C.$ The duality carries over to lower bounds for $\gamma(C;S,S')$ and $\gamma(-C;S',S)$
that are obtained with Theorem \ref{T:MTS}.
 

\begin{lemma}
For a given divisor $C$, and for a sequence of divisors $\{A_i \}$ as in Theorem \ref{T:MTS}, let
\[
\gamma(C;S,S')  ~\geq~ |\Delta \cap I'| + |\Delta' \cap I| - |\Delta'|.
\]
Then
\[
\gamma(-C;S',S) ~\geq~ |\Delta' \cap I| + |\Delta \cap I'|-|\Delta|.
\]
Moreover, for a long enough seqeunce such that $\deg A_0 < \min \{ 0, \deg C \}$ and $\deg A_n > \max \{ 2g-2, 2g-2+\deg C \}$, the difference
between the two lower bounds $|\Delta| - |\Delta'| = \deg C.$
\end{lemma}

\begin{proof}
To obtain the bound for $\gamma(-C;S',S)$ we apply the theorem with the sequence $\{ A_i-C \}$. This exchanges $\Delta$ and $\Delta'$, and $I$ and $I'$.
The second claim reduces to the following statement:
\begin{align*}
|\Delta|-|\Delta'| &= |\{i: A_i \in \Gamma_{P_i}\}| - |\{i: A_i - C \in \Gamma_{P_i}\}|  \\
                   &= (l(A_n)-l(A_0)) - (l(A_n-C)-l(A_0-C))   \\
                   &= (l(A_n)-l(A_n-C)) - (l(A_0)-l(A_0-C)). 
\end{align*}
For divisors $A_0$ and $A_n$ in the give range, the last difference equals $\deg C.$
\end{proof}

Note that for an arbitrary sequence $\{ A_i \}$ and for $C = C^+ - C^{-}$, where $C^+, C^- \geq 0$, the proof indicates that 
$|\Delta|-|\Delta'| \leq \deg C^+ + \deg C^-.$ In general we expect the lower bound for $\gamma(C;S,S')$ to increase when $S$ and $S'$ are 
enlarged. On the other hand, for an effective divisor $C$ wihtout base points, $C \in \Gamma(C;S,S')$ and $\gamma(C;S,S') = \deg C$, for all
$S$ and $S'$. For an arbitrary effective divisor $C$, we show that Theorem \ref{T:MTS} yields the best results 
when $S$ contains the base points of $C$.  

\begin{lemma}
For a given effective divisor $C$ and set $S'$, and for any sequence $\{ A_i \}$, the lower bound in Theorem \ref{T:MTS} 
attains its maximum for $S$ equal to the set of base points of $C$.
\end{lemma}

\begin{proof}
Clearly, for any sequence $\{ A_i \}$, the set $S$ is optimal if it contains $\{P_i : i \in \Delta'\}$. 
For $i \in \Delta'$, $A_i \not \in \Gamma_{P_i}$ and
$A_i - C \in \Gamma_{P_i}$. The semigroup property of $\Gamma_{P_i}$ implies that $C \not \in \Gamma_{P_i}$. For an effective divisor $C$ there is no gain 
in assuming that $S$ contain points other than the basepoints of $C$.
\end{proof}

\section{Order bounds in semigroup form} \label{S:order}


In this section we prove the order bounds $d_{DK}, d_{DP},$ and $d_B$ using a combination of Theorem \ref{T:S'} and Theorem \ref{T:MTS}. 
To obtain lower bounds for the minimum distance $d$ of an AG code, we use $d \geq \gamma(C;S,\emptyset)$ (Lemma \ref{L:codegammas}) and estimate
$\gamma(C;S,\emptyset)$, where $C$ is the designed minimum support of the code and the code is defined with divisor $D$ disjoint from $S$.  
Theorem \ref{T:MTS} gives us a way to obtain lower bounds for $\gamma(C;S,S')$ but the lower bounds are nontrivial only if $S' \neq \emptyset$. This is where
we use Theorem \ref{T:S'}. We have
\[
\Gamma(C;S,\emptyset) = \bigcup_{\lambda \in \Lambda'} \Gamma(C+\lambda;S,S'),
\]
where $\Lambda'$ is the semigroup generated by the points in $S'$. Now Theorem \ref{T:MTS} can be used to estimate $\gamma(C+\lambda;S,S')$, 
for $\lambda \in \Lambda'.$ 

\begin{theorem} 
(The bound $d_{DK}$ \cite{DuuKir09}) 
\label{T:DK}
Let $C$ be a divisor and let $S$ be a finite set of rational points. For any finite set $S'$ of rational points,
\[
\gamma(C;S,\emptyset) = \min_{\lambda \in \Lambda'} \gamma(C+\lambda;S,S') \geq \min_{\lambda \in \Lambda'} \gamma_\ast(C+\lambda;S,S'),
\] 
where $\Lambda'$ is the semigroup generated by the points in $S'$ and $\gamma_\ast(C+\lambda;S,S')$ is a lower bound for 
$\gamma(C+\lambda;S,S').$
\end{theorem}

It is helpful to interpret the data in the theorem as a directed graph with vertices a collection ${\cal C}$ of divisors $C$ and edges $(C,C+Q)$, for
$C \in {\cal C}$, $Q \in S'$. If we label the vertex $C \in {\cal C}$ with $\gamma(C;S,S')$ then $\gamma(C;S,\emptyset)$ is the minimum of all vertex labels 
$\gamma(C';S,S')$ for $C' \geq C$.
Among the estimates $\gamma_B, \gamma_{DP}$ and $\gamma_{DK}$ for $\gamma(C+\lambda;S,S')$ obtained with Theorem \ref{T:MTS}, 
only $\gamma_{DK}$ uses sets $S'$ of size larger than one. For the other two types we use 
\[
\Gamma(C+\lambda;S,S') = \bigcap_{Q \in S'} \Gamma(C+\lambda;S,Q)
\]
in combination with estimates for $\gamma(C+\lambda;S,Q).$


\begin{corollary} \label{C:BDP}
(The bounds $d_B$ \cite{Bee07} and $d_{DP}$ \cite{DuuPar09} in semigroup form) 
\label{T:BDP}
\[
\gamma(C;S,\emptyset) \geq \min_{\lambda \in \Lambda'}\,(\max_{Q \in S'} \gamma_\ast(C+\lambda;S,Q)\,),
\] 
where $\gamma_\ast(C+\lambda;S,Q)$ is a lower bound for $\gamma(C+\lambda;S,Q).$
\end{corollary}

\begin{proof}
\[
\gamma(C+\lambda;S,S') = \max_{Q \in S'} \gamma(C+\lambda;S,Q) \geq \max_{Q \in S'} \gamma_\ast(C+\lambda;S,Q)
\]
\end{proof}

For an interpretation of the corollary in graph terms we assign a label $\gamma(C;S,Q)$ to each edge $(C,C+Q)$ and then label the vertex 
$C$ with the maximum of the labels on the outgoing edges $(C,C+Q)$, for $Q \in S'$.
The difference between the bounds $d_{B}$ and $d_{DP}$ is not in Corollary \ref{T:BDP} but in the way that each uses Theorem \ref{T:MTS}
to obtain the lower bounds $\gamma_\ast(C+\lambda;S,Q)$.

\begin{example} For $C=-5P+8Q$, we estimate $\gamma(C;\{P,Q\},\emptyset)$ in two different ways.
From Example \ref{E:58}, the labels for the edges $(C,C+P)$ and $(C,C+Q)$ are
\[
\gamma_{DP}(-5P+8Q;\{P,Q\},P) = \gamma_{DP}(-5P+8Q;\{P,Q\},Q) = 6.
\]
The estimates are critical in Corollary \ref{C:BDP} which yields $\gamma(C;\{P,Q\},\emptyset) \geq 6.$ 
On the other hand, a direct estimate of the vertex label at $C$ gives
\[
\gamma_{DK}(-5P+8Q;\{P,Q\},\{P,Q\})) = 7. 
\]
And Theorem \ref{T:DK} yields $\gamma(C;\{P,Q\},\emptyset) \geq 7.$ 
\end{example}


\section{Order bounds in sequence form} \label{S:orderpath} 

The bounds $d_{B}$ and $d_{DP}$ in Corollary \ref{T:BDP} use Theorem \ref{T:S'} and differ from their original formulation, which is based 
on repeated use of Proposition \ref{P:rec}.
\[
\Gamma(C;S,\emptyset) = \Gamma(C;S,Q) \cup \Gamma(C+Q;S,\emptyset).
\] 
In this section we compare the different formulations and show that they are in agreement. 


\begin{proposition} (The bounds $d_B$ \cite{Bee07} and $d_{DP}$ \cite{DuuPar09} in sequence form) \label{P:path}
Let $C$ be a divisor and let $S$ be a finite set of rational points. For any subset $S'$ of rational points, and 
for a long enough sequence of points $Q_0, \Q_1, \ldots, Q_r \in S'$,
\[
\gamma(C;S,\emptyset) \geq \min_{j=0,\ldots,r}  \gamma_\ast(C+R_j;S,Q_j) .
\]
Here $R_0=0$ and $R_j = R_{j-1}+Q_{j-1},$ for $j > 0,$ and $\gamma_\ast(C+R_j;S,Q_j)$ is a lower bound for $\gamma(C+R_j;S,Q_j)$.
\end{proposition}

\begin{proof} With Proposition \ref{P:rec}, 
\[
\Gamma(C;S,\emptyset ) = \cup_{j=0,\ldots,r} \Gamma(C+R_j;S,Q_j) \cup \Gamma(C+R_r+Q_r;S,\emptyset).
\]
\end{proof}

As before, Theorem \ref{T:MTS} can be used to estimate $\gamma(C+R_j;S,Q_j)$, for $j=0,1,\ldots,r.$
Extending the graph interpretation for the bounds $d_B$ and $d_{DP}$ given after Corollary \ref{C:BDP},
we interpret the label $\gamma(C+R_j;S,Q_j)$ for the edge $(C+R_j,C+R_j+Q_j)$ as the flow capacity along
the edge. The order bound in sequence form estimates $\gamma(C;S,\emptyset)$ as the maximum flow 
capacity of any long enough path $(C,C+Q_0,C+Q_0+Q_1,\ldots)$. The order bound in \cite{CMST07} estimates 
the labels $\gamma(C+R_j;S,Q_j)$ in the same way as the Beelen bound but assigns a special point $P \in S'$ and 
computes the maximum flow along a path $(C,C+P,C+2P,\ldots)$ with $Q_0 = Q_1 = \cdots = Q_r = P.$ 

\begin{example}
The code $C_\Omega(D,K+9P+Q)$, defined with the Suzuki curve over $\ff_8$, has designed minimum support $C=9P+Q$ and designed minimum distance $d_{GOP} = 10$.
For $D$ disjoint form $P$ and $Q$, the actual distance of the code is at least $13$. To see this using the Beelen
bound it is important to choose $Q_0 = P$ and $Q_1 = Q_2 = Q$. The constant choices $Q_0 = Q_1 = Q_2 = P$ and
$Q_0 = Q_1 = Q_2 = Q$ yield only $d \geq 11$ and $d \geq 12$, respectively.
\begin{align*}
&\min \, \{ \gamma_B(9P+Q;P,P),  \gamma_B(10P+Q;Q,Q), \gamma_B(10P+2Q;Q,Q) \} = \min \, \{ 13, 13, 14 \} = 13. \\
&\min \, \{ \gamma_B(9P+Q;P,P),  \gamma_B(10P+Q;P,P), \gamma_B(11P+Q;P,P) \} = \min \, \{ 13, 11, 14 \} = 11. \\
&\min \, \{ \gamma_B(9P+Q;Q,Q),  \gamma_B(9P+2Q;Q,Q), \gamma_B(9P+3Q;Q,Q) \} = \min \, \{ 12, 13, 13 \} = 12.
\end{align*}
\end{example}

In general, $\Gamma(C+P;S,Q) \subseteq \Gamma(C;S,Q)$ for $P \neq Q$, and thus
$\gamma(C+P;S,Q) \geq \gamma(C;S,Q)$. Therefore, if $\gamma_\ast(C+P;S,Q)$ and $\gamma_\ast(C;S,Q)$ are lower bounds, then we can assume
that $\gamma_\ast(C+P;S,Q) \geq \gamma_\ast(C;S,Q)$, for otherwise we would replace $\gamma_\ast(C+P;S,Q)$ with $\gamma_\ast(C;S,Q)$.
With this assumption, the bounds in Corollary \ref{C:BDP} and Proposition \ref{P:path} agree.
 

\begin{proposition} \label{P:agree}
Let $\{ \gamma_\ast(C+\lambda;S,Q) : \lambda \in \Lambda', Q \in S' \}$ be a collection of lower bounds for the
corresponding set of actual values $\{ \gamma(C+\lambda;S,Q) \}$ such that the estimates satisfy $\gamma_\ast(C+\lambda+P;S,Q) \geq \gamma_\ast(C+\lambda;S,Q)$
whenever $P \neq Q$. Then
\[
\max_{Q_0, Q_1, \ldots, Q_r \in S'} \,(\, \min_{j=0,\ldots,r} \; \gamma_\ast(C+R_j;S,Q_j)  \,) ~=~ \min_{\lambda \in \Lambda'} \, ( \, \max_{Q \in S'} \; \gamma_\ast(C+\lambda;S,Q) \, ).
\]
\end{proposition}

\begin{proof}
The two sides of the equality represent lower bounds for $\gamma(C;S,\emptyset)$ obtained with Proposition \ref{P:path} and
Corollary \ref{C:BDP}, respectively. Denote the left side by $\gamma_{seq}$ and the right sight by $\gamma_{sgp}$. 
Clearly, $\gamma_{seq} \geq \gamma_{sgp}$ and it suffices to show that $\gamma_{sgp} \geq \gamma_{seq}.$
Assume that there exists $\lambda \in \Lambda'$ with
$\max_{Q \in S'} \gamma_\ast(C+\lambda;S,Q) < \gamma_{seq}.$ Using $\gamma(C;S,Q) \leq \gamma(C+P;S,Q)$ for $P \neq Q$,
we see that $\gamma(C+\lambda';S,Q) < \gamma_{seq}$ for all $\lambda_Q \leq \lambda' \leq \lambda,$ where $\lambda_Q$ 
is the $Q-$component of $\lambda$. Every long enough path $R_0, R_1, R_2, \ldots$ contains some $R \leq \lambda$ with $R_Q = \lambda_Q$ for some $Q$.
But then $\lambda_Q \leq R \leq \lambda$ and $\gamma(C+R;S,Q) < \gamma_{seq}$, a contradiction.
\end{proof}

In Proposition \ref{P:path}, it is not clear 
how to choose an optimal sequence $Q_0,Q_1, \ldots, Q_r.$ It follows from Proposition \ref{P:agree} that, 
once it has been decided to choose the $Q_i$ from a finite set $S'$, 
the choice of an optimal seqeunce can be made in a straightforward way, namely by following a greedy procedure as follows: 
For a sequence starting with $Q_0, Q_1, \ldots, Q_{i-1}$, choose $Q_i \in S'$ such that the edge label $\gamma_\ast(C+R_i;S,Q_i)$
is maximal among $\gamma_\ast(C+R_i;S,Q),$ for $Q \in S'.$ 

\begin{corollary} \label{C:greedy} 
The lower bound in Proposition \ref{P:path} is optimal for a choice of $Q_j$, $j=0,1,\ldots,r$, such that 
$\gamma_\ast(C+R_j;S,Q_j) = \max_{Q \in S'} \gamma_\ast(C+R_j;S,Q).$ 
\end{corollary}

\begin{proof}
The choice gives a lower bound $\gamma_{seq,greedy}$ satisfying $\gamma_{seq} \geq \gamma_{seq,greedy} \geq \gamma_{sgp}$.
In Proposition \ref{P:agree} it was shown that $\gamma_{seq} = \gamma_{sgp}$ and therefore also $\gamma_{seq} = \gamma_{seq,greedy}.$
\end{proof}


\section{Computing the lower bounds} \label{S:comp}

We present computational short-cuts that make it feasible to establish the various bounds in the paper for 
large numbers of codes from a given curve whose geometry is well understood. 
For two-point codes from Hermitian curves, Suzuki curves and Giulietti-Korchmaros curves, numerical results are 
available in interactive form at \cite{Kir09}. The comparison Table \ref{Table:comparison} gives a summary of the
results for two-point codes on the Suzuki curves over $\ff_8$ and $\ff_{32}.$  
The Suzuki curve over $\ff_8$ has genus $g = 14$. For a given degree there are $m = 13$ two-point codes. 
For a designed distance in the range $0, 1, \ldots, 2g-1 = 27$ there are $2g \cdot m = 364$ two-point codes. 
For the Suzuki curve over $\ff_{32}$ the numbers are $g=124$ and
$m=41$ for a total of $2g \cdot m = 10168$ two-point codes.  

\begin{table}[t]
\begin{center}
\begin{tabular}{cc}
Suzuki over $\ff_8$ & Suzuki over $\ff_{32}$\\
\begin{tabular}{lrrrrr}
\toprule
& & $d_{LM}$ & $d_{ABZ}$ &$d_{B}$ & $d_{DK}$ \\
\midrule
$d_{GOP}$ & & 228 & 228 & 228 & 228\\
$d_{LM}$ & & 0 & 29 & 102 & 108 \\
$d_{ABZ}$ & & 0 & 0 & 94 & 98 \\
$d_{B}$ & & 1 & 3 & 0 & 15 \\
\midrule
$d_{GOP}$ & & 4 & 5 & 6 & 6\\
$d_{LM}$ & & 0 & 1 & 4 & 4 \\
$d_{ABZ}$ & & 0 & 0 & 4 & 4 \\
$d_{B}$ & & 1 & 1 & 0 & 1 \\
\bottomrule
\end{tabular} &
\begin{tabular}{lrrrrr}
\toprule
& & $d_{LM}$ & $d_{ABZ}$ &$d_{B}$ & $d_{DK}$ \\
\midrule
$d_{GOP}$ & & 6352 & 6352 & 6352 & 6352\\
$d_{LM}$ & & 0 & 2852 & 4729 & 4757 \\
$d_{ABZ}$ & & 0 & 0 & 4683 & 4711 \\
$d_{B}$ & & 1 & 1 & 0 & 1565 \\
\midrule
$d_{GOP}$ & & 8 & 21 & 33 & 33\\
$d_{LM}$ & & 0 & 15 & 28 & 28 \\
$d_{ABZ}$ & & 0 & 0 & 24 & 24 \\
$d_{B}$ & & 1 & 1 & 0 & 6 \\
\bottomrule
\end{tabular}\\
\end{tabular}
\end{center}
\caption{Comparison of bounds for $364$ Suzuki codes over $\ff_8$ ($g=14$) and for $10168$ Suzuki codes over $\ff_{32}$ ($g=124$). Number of improvements of one bound over another (top), and the maximum improvement (bottom).}
\label{Table:comparison}
\end{table}


\subsection{Floor bounds}

If a floor bound is to be used for a code with designed minimum support $C$ a choice of auxiliary divisors is needed, such as the divisors $A$ and $B$ 
in the $ABZ$ bounds. In the generic case it is not clear how to choose divisors that produce the best bound. A natural approach is to choose $C$ with
support in a small set of points and to choose $A$ and $B$ among all divisors with support in those points. Important special cases are one-point codes
with $A, B$ and $C$ supported in a point $P$, and two-point codes with $A, B$ and $C$ supported in points $P$ and $Q$. In general let $C$ belong to a 
family of divisors ${\cal C}$ and $A$ to a family of divisors ${\cal A}$. For the efficient optimization we use that ${\cal A}$ has a natural partial
ordering such that $A' \leq A$ if $A-A'$ is effective. 
For each of the bounds $d_{ABZ}$, $d_{GST}$, and $d_{LM}$, we first build a table with the dimension $l(A)$ of the Riemann-Roch space $L(A)$,
for all $A \in {\cal A}$. When ${\cal A}$ consists of divisors supported in a point $P$ or in points $\{P,Q\}$ this essentially asks for the Weierstrass nongaps, 
either for one-point divisors or more generally for two-point divisors. For Hermitian and Suzuki curves, two-point nongaps are known in closed form 
\cite{Mat01}, \cite{BeeTut06}, \cite{DuuPar08}. Parsing though all two-point divisors 
in increasing degree order we update $l(A)$ knowing $l(A-P)$ and whether there is a $P$-gap at $A$. For the bounds $d_{GST}$ and $d_{LM}$ we also store
the floor $\lfloor A \rfloor$ for each $A \in {\cal A}$. For a given divisor $C$, the bounds can then be computed as follows. \\

The bound $d_{ABZ}$ (Theorem \ref{T:ABZ}):~~
For given $C$, compute $f(A) = l(A)-l(A-C)$ for all $A \in {\cal A}$ in increasing order. For each $A$ keep track of 
the quantity $F(A)=\max_{A' \leq A} f(A')$ and update $d_{ABZ}$ with the greater of $d_{ABZ}$ and $\deg C + F(A) - f(A)$. \\

The bound $d_{GST}$ (Corollary \ref{C:gst}, Theorem \ref{T:GST1p}):~  
For given $C$, compute $f(A) = l(A)-l(A-C)$ for all $A \in {\cal A}$ in increasing order. For each $A$ 
update $d_{GST}$ with the greater of $d_{GST}$ and $\deg C + f(\lfloor A \rfloor) - f(A)$. \\


The bound $d_{LM}$ (Corollary \ref{C:LM}):~
For given $C$, compute $f(A) = l(A)-l(A-C)$ for all $A \in {\cal A}$ in increasing order. For each $A$ 
and for all $\lfloor A \rfloor \leq A' \leq A$ such that $f(A')-f(A) = \deg A - \deg A'$ update $d_{LM}$ with 
the greater of $d_{LM}$ and $\deg C + f(A') - f(A)$. \\

Pairs $\lfloor A \rfloor \leq A' \leq A$ such that $f(A')-f(A) = \deg A - \deg A'$ satisfy $L(A) = L(A')$ and $L(K+C-A)=L(K+C-A')$. When $A, A'$ are chosen
from a two-point family ${\cal A} = \{mP+nQ\}$ the search over such pairs can be optimized as follows. As part of the precompution we build a type of one dimensional ceiling divisor, that is a function $cl(A)$ returning the maximum $a$ for which $l(A)=l(A+aP)$. For each non-negative $b$ with $l(K+C-A+bQ)=l(K+C-A)$ we read off a cooresponding $a=cl(K+C-A+bQ)$ and then update $d_{LM}$ with the greater of $d_{LM}$ and $\min\{a,fl_P\}+\min\{b,fl_Q\}$ where $fl_P=(A-\lfloor A \rfloor)_P$ and $fl_Q=(A-\lfloor A \rfloor)_Q$.

\subsection{Order bounds}

Order bounds for estimating the minimum distance of a given code have two steps. For a code with designed minimum support $C$ and divisor $D$ disjoint from $S$,
the minimum distance is at least $\gamma(C;S,\emptyset)$. First the main theorem (Theorem \ref{T:MTS}) is used to obtain lower bounds for $\gamma(C+\lambda;S,S')$, 
for effective divisors $\lambda$ with support in $S'$. Then Theorem \ref{T:S'} combines the lower bounds into a lower bound for $\gamma(C;S,\emptyset).$ 
By the nature of the order bound, the estimates in the first step can be used to obtain lower bounds for subcodes of the given code. When computing order bounds
we therefore fix a partially ordered family ${\cal C}$ of divisors $C$ and simultaneously estimate the distance for all divisors $C \in {\cal C}$. 
In practice we have used families of two-point divisors of absolute degree $|\deg C| \leq 2g-1$.  \\


Order bound $d_{DK}$ (Theorem \ref{T:DK}):~ For each $C \in {\cal C}$, in decreasing order, compute $\gamma_{DK}(C;S,S')$, and let
$d_{DK}(C)$ be the smaller of $\min_{Q \in S'} d_{DK}(C+Q)$ and $\gamma_{DK}(C;S,S')$. \\

Order bounds $d_{DP}, d_B$ (Corollary \ref{C:BDP}, Proposition \ref{P:path}):~ For each $C \in {\cal C}$, in decreasing order, 
compute $\gamma_{\ast}(C;S,Q)$, for $Q \in S'$, and let $d_{\ast}(C) = \max_{Q \in S'} \{ \min ( d_{\ast}(C+Q), \gamma_{\ast}(C;S,Q) ) \}$. \\ 

To estimate $\gamma(C;S,S')$ (or $\gamma(C;S,Q)$) for a fixed $C$ using Theorem \ref{T:MTS}, 
we need to choose a sequence of divisors $A_i$. It is not clear in general how to choose 
a sequence that produces the best bound. We choose the sequence $A_i$ inside a given family ${\cal A}$ and represent the divisors in ${\cal A}$ as a directed 
grid graph where the divisors $A_i$ are the vertices and edges $(A_{i-1}, A_i)$ correspond to pairs $A_i = A_{i-1}+P_i$, with $P_i$ a rational point.
On such a graph we label the edges with $0$ or $1$ according to whether the estimate in Theorem \ref{T:MTS} increases when we follow the particular edge. 
Using a graph path maximizing algorithm we can find the best bound for $\gamma(C;S,S')$ as a path with the most ones in one run through the graph. When the
family ${\cal A}$ is the family of all two-point divisors $\{ mP+nQ \}$, the graph is a rectangular grid. In that case, the bound $d_{DK}$ optimizes over all
paths in the grid. The bound $d_{DP}$ optimizes over all paths but only considers labels in one direction (say the $P$ direction), ignoring the possible gains 
along edges in the other direction (the $Q$ direction). Finally the bound $d_B$ selects an optimal straight path in the grid. \\

To keep track of the estimates in the order bound we use a directed grid graph with vertices $C \in {\cal C}$, as in Sections \ref{S:order} and \ref{S:orderpath}. 
For each vertex $C \in {\cal C}$ we consider the graph with vertices $A \in {\cal A}$ and edges labeled with $0$ or $1$ as described above. A path maximizing 
algorithm for the graph on ${\cal A}$ yields either $\gamma(C;S,S')$ (for order bounds in semigroup form) or $\gamma(C;S,Q)$ (for order bounds in sequence form). 
For order bounds in semigroup form, we label the vertex $C \in {\cal C}$ with $\gamma(C;S,S')$ and compute $\gamma(C;S,\emptyset)$ as the minimum of all labels $\gamma(C';S,S')$ for $C' \geq C$ (Theorem \ref{T:DK}). For order bounds in sequence form, we label
the edge $(C,C+Q)$ with $\gamma(C;S,Q)$. If we interpret the label as the flow capacity along the edge then $\gamma(C;S,\emptyset)$ is the maximum flow capacity 
of any long enough path $(C,C+Q_0,C+Q_0+Q_1,\ldots)$ in the graph (Proposition \ref{P:path}). For the order bound in 
sequence form we may label the vertices $C \in {\cal C}$ with the maximum of the labels on the outgoing edges and then apply vertex minimization. By Proposition
\ref{P:agree} this results in the same bound. Also, the labeling of the edges in the graph is such that a path of maximum flow can be found efficiently
in a greedy way: At every vertex $C$ continue the path along an edge $(C,C+Q)$ of maximum flow capacity. By Corollary \ref{C:greedy} this results again
in the same bound.  



\subsection{Examples}

Table \ref{T:Suz8} gives a selection of two-point codes and their bounds for the Suzuki curve over $\ff_8$. Codes are included to illustrate differences
between bounds and to compare with known results. To select optimal codes we recommend using the tables \cite{Kir09}. The top part of the table lists all codes
with $d_{GST} > d_{LM}$ and extends Table 1 in \cite{GST09} (the entries with footnote $1$). The middle part lists the remaining codes with $d_{GST2} > d_{LM}$
and extends Table 2 in \cite{GST09} (the entries with footnote $2$). The bound $\tilde d$ refers to examples in \cite[Table 3]{GST09}. 
Columns $A$ and $B$ list divisors that optimize $d_{ABZ+}$. A footnote $+$ indicates that
the choice is optimal for $d_{ABZ+}$ but not for $d_{ABZ}$. A footnote $f$ indicates that the choice is optimal for $d_{ABZ}$ after $A$ and $B$ are replaced with
their floors $\lfloor A \rfloor$ and $\lfloor B \rfloor$, respectively. All other choices simultanously optimize $d_{ABZ}$ and $d_{ABZ+}$.


\begin{table}[th]
\begin{center}
\begin{tabular}{l@{\hspace{8mm}}r@{\hspace{3mm}}r@{\hspace{8mm}}r@{\hspace{3mm}}r@{\hspace{8mm}}r@{\hspace{3mm}}r@{\hspace{8mm}}r@{\hspace{3mm}}r@{\hspace{8mm}}r@{\hspace{3mm}}r@{\hspace{8mm}}r@{\hspace{3mm}}rrr}
\toprule
& & &\multicolumn{2}{l}{$d_{GOP}$}   &\multicolumn{2}{l}{$d_{GST}$}   &\multicolumn{2}{l}{$d_{GST2}$}   &\multicolumn{2}{l}{$d_{B}$} &\multicolumn{2}{l}{$\tilde d$} \\
\multicolumn{1}{l}{$G$} &A &B & &\multicolumn{2}{l}{$d_{LM}$}   &\multicolumn{2}{l}{$d_{ABZ}$}  &\multicolumn{2}{l}{$d_{ABZ+}$} 
&\multicolumn{2}{l}{$d_{ABZ'}$}   &\multicolumn{2}{l}{$d_{DK}$}  \\
\midrule
$(22,4)^{1,f}$ &14P &8P &0 &3 &4 &4 &3 &4 &5 &5 &- &5 \\
$(21,5)^{1,+}$ &13P &8P &0 &3 &4 &4 &3 &5 &5 &5 &- &5 \\
$(20,6)^f$   &14P &6P   &0 &4 &5 &5 &4 &5 &6 &6 &- &6 \\
(20,7)   &14P &6P       &1 &4 &5 &5 &4 &5 &6 &6 &- &6 \\
$(23,4)^f$ &15P &8P     &1 &4 &5 &5 &4 &5 &6 &6 &- &6 \\
(21,6)   &13P &8P       &1 &4 &5 &5 &4 &5 &6 &6 &- &7 \\
(22,6) &14P &8P         &2 &5 &6 &6 &5 &6 &7 &7 &- &7 \\
$(24,4)^{1,2}$ &16P &8P &2 &4 &5 &5 &5 &6 &6 &6 &- &6 \\
$(24,5)^2$ &16P &8P     &3 &5 &6 &6 &6 &7 &7 &7 &- &7 \\
$(24,6)^{1,2}$ &16P &8P &4 &6 &7 &7 &7 &7 &7 &7 &- &7 \\
 (26,4) &16P &10P       &4 &6 &7 &7 &6 &7 &8 &8 &- &8 \\
\midrule
$(24,3)^2$ &14P &10P    &1 &3 &3 &3 &4 &5 &6 &6 &- &6 \\
(27,0) &13P &13P        &1 &2 &2 &2 &3 &4 &6 &6 &- &6 \\
$(30,1)^2$ &13P &13P    &5 &7 &7 &8 &8 &8 &8 &8 &8 &8 \\
$(32,1)^2$ &13P &13P    &7 &9 &9 &10 &10 &10 &10 &10 &10 &10 \\
(40,0) &26P &13P        &14 &15 &15 &15 &16 &16 &16 &16 &- &16 \\
\midrule
$(24,2)^{+}$ &16P &8P &0 &3 &3 &3 &3 &3 &4 &4 &4 &4 \\
$(25,1)^{+}$ &13P &12P &0 &2 &2 &2 &2 &3 &6 &6 &- &6 \\
(21,7) &13P &8P &2 &5 &5 &5 &5 &5 &6 &6 &- &7 \\
(21,8) &13P &8P &3 &5 &5 &5 &5 &5 &6 &6 &- &7 \\
(27,1) &13P &13P &2 &4 &4 &4 &4 &6 &7 &8 &6 &8 \\
(28,1) &13P &13P &3 &6 &6 &6 &6 &6 &7 &8 &8 &8 \\
(29,1) &13P &13P &4 &6 &6 &6 &6 &8 &8 &8 &8 &8 \\
(28,2) &13P &13P &4 &8 &8 &8 &8 &8 &7 &8 &8 &8 \\
(30,2) &13P &13P &6 &9 &9 &10 &9 &10 &9 &10 &10 &10 \\
(30,3) &13P &13P &7 &9 &9 &10 &9 &10 &10 &10 &- &10 \\
$(31,1)^{+}$ &21P &10P &6 &8 &8 &8 &8 &8 &9 &9 &9 &9 \\
$(33,1)^{+}$ &23P &10P &8 &10 &10 &10 &10 &10 &11 &11 &11 &11 \\
$(33,3)^{+}$ &23P &10P &10 &12 &12 &12 &12 &12 &12 &12 &- &13 \\
$(34,3)$ &24P &10P &11 &12 &12 &12 &12 &12 &12 &12 &- &13 \\
\bottomrule
\end{tabular}
\end{center}
\caption{Selected two-point codes on the Suzuki curve over $\ff_8$}
\label{T:Suz8}
\end{table}




\bibliographystyle{plain}

\begin{thebibliography}{10}

\bibitem{AndGei08}
Henning~E. Andersen and Olav Geil.
\newblock Evaluation codes from order domain theory.
\newblock {\em Finite Fields Appl.}, 14(1):92--123, 2008.

\bibitem{Bee07}
Peter Beelen.
\newblock The order bound for general algebraic geometric codes.
\newblock {\em Finite Fields Appl.}, 13(3):665--680, 2007.

\bibitem{BeeTut06}
Peter Beelen and Nesrin Tuta{\c{s}}.
\newblock A generalization of the {W}eierstrass semigroup.
\newblock {\em J. Pure Appl. Algebra}, 207(2):243--260, 2006.

\bibitem{CMST07}
C{\'{\i}}cero Carvalho, Carlos Munuera, Ercilio da~Silva, and Fernando Torres.
\newblock Near orders and codes.
\newblock {\em IEEE Trans. Inform. Theory}, 53(5):1919--1924, 2007.

\bibitem{CarTor05}
C{\'{\i}}cero Carvalho and Fernando Torres.
\newblock On {G}oppa codes and {W}eierstrass gaps at several points.
\newblock {\em Des. Codes Cryptogr.}, 35(2):211--225, 2005.

\bibitem{DuuKir09}
Iwan Duursma and Radoslav Kirov.
\newblock An extension of the order bound for ag codes.
\newblock In {\em Applied algebra, algebraic algorithms and error-correcting
  codes}, volume 5527 of {\em Lecture Notes in Comput. Sci.}, pages 11--22.
  Springer, Berlin, 2009.

\bibitem{DuuPar09}
Iwan Duursma and Seungkook Park.
\newblock {Coset bounds for algebraic geometric codes,
  doi:10.1016/j.ffa.2009.11.006}.
\newblock {\em {Finite Fields Appl.}}

\bibitem{DuuPar08}
Iwan Duursma and Seungkook Park.
\newblock {Coset bounds for algebraic geometric codes, extended version}.
\newblock arXiv:0810.2789, 2008.

\bibitem{Duu93}
Iwan~M. Duursma.
\newblock Majority coset decoding.
\newblock {\em IEEE Trans. Inform. Theory}, 39(3):1067--1070, 1993.

\bibitem{FenRao93}
Gui~Liang Feng and T.~R.~N. Rao.
\newblock Decoding algebraic-geometric codes up to the designed minimum
  distance.
\newblock {\em IEEE Trans. Inform. Theory}, 39(1):37--45, 1993.

\bibitem{GKL93}
Arnaldo Garc{\'{\i}}a, Seon~Jeong Kim, and Robert~F. Lax.
\newblock Consecutive {W}eierstrass gaps and minimum distance of {G}oppa codes.
\newblock {\em J. Pure Appl. Algebra}, 84(2):199--207, 1993.

\bibitem{GarLax92}
Arnaldo Garc{\'{\i}}a and R.~F. Lax.
\newblock Goppa codes and {W}eierstrass gaps.
\newblock In {\em Coding theory and algebraic geometry ({L}uminy, 1991)},
  volume 1518 of {\em Lecture Notes in Math.}, pages 33--42. Springer, Berlin,
  1992.

\bibitem{GST09}
Cam G\"uneri, Henning Sitchtenoth, and Ishan Taskin.
\newblock Further improvements on the designed minimum distance of algebraic
  geometry codes.
\newblock {\em J. Pure Appl. Algebra}, 213(1):87--97, 2009.

\bibitem{HoeLinPel98}
Tom H{\o}holdt, Jacobus~H. van Lint, and Ruud Pellikaan.
\newblock Algebraic geometry of codes.
\newblock In {\em Handbook of coding theory, Vol. I, II}, pages 871--961.
  North-Holland, Amsterdam, 1998.

\bibitem{HomKim01}
Masaaki Homma and Seon~Jeong Kim.
\newblock Goppa codes with {W}eierstrass pairs.
\newblock {\em J. Pure Appl. Algebra}, 162(2-3):273--290, 2001.

\bibitem{HomKim06}
Masaaki Homma and Seon~Jeong Kim.
\newblock The complete determination of the minimum distance of two-point codes
  on a {H}ermitian curve.
\newblock {\em Des. Codes Cryptogr.}, 40(1):5--24, 2006.

\bibitem{KirPel95}
Christoph Kirfel and Ruud Pellikaan.
\newblock The minimum distance of codes in an array coming from telescopic
  semigroups.
\newblock {\em IEEE Trans. Inform. Theory}, 41(6, part 1):1720--1732, 1995.
\newblock Special issue on algebraic geometry codes.

\bibitem{Kir09}
Radoslav Kirov.
\newblock Parameters on algebraic geometric codes,
  \url{http://agtables.appspot.com}.

\bibitem{LunMcc06}
Benjamin Lundell and Jason McCullough.
\newblock A generalized floor bound for the minimum distance of geometric
  {G}oppa codes.
\newblock {\em J. Pure Appl. Algebra}, 207(1):155--164, 2006.

\bibitem{MahMat06}
Hiren Maharaj and Gretchen~L. Matthews.
\newblock On the floor and the ceiling of a divisor.
\newblock {\em Finite Fields Appl.}, 12(1):38--55, 2006.

\bibitem{MMP05}
Hiren Maharaj, Gretchen~L. Matthews, and Gottlieb Pirsic.
\newblock Riemann-{R}och spaces of the {H}ermitian function field with
  applications to algebraic geometry codes and low-discrepancy sequences.
\newblock {\em J. Pure Appl. Algebra}, 195(3):261--280, 2005.

\bibitem{Mat01}
Gretchen~L. Matthews.
\newblock Weierstrass pairs and minimum distance of {G}oppa codes.
\newblock {\em Des. Codes Cryptogr.}, 22(2):107--121, 2001.

\bibitem{Mir95}
Rick Miranda.
\newblock {\em Algebraic curves and {R}iemann surfaces}, volume~5 of {\em
  Graduate Studies in Mathematics}.
\newblock American Mathematical Society, Providence, RI, 1995.

\bibitem{Par09}
Seungkook Park.
\newblock Minimum distance of hermitian two-point codes.
\newblock {\em Des. Codes Cryptogr.}, To appear.

\bibitem{TsfVla07}
Michael Tsfasman, Serge Vl{\u{a}}du{\c{t}}, and Dmitry Nogin.
\newblock {\em Algebraic geometric codes: basic notions}, volume 139 of {\em
  Mathematical Surveys and Monographs}.
\newblock American Mathematical Society, Providence, RI, 2007.

\bibitem{LinWil86}
Jacobus~H. van Lint and Richard~M. Wilson.
\newblock On the minimum distance of cyclic codes.
\newblock {\em IEEE Trans. Inform. Theory}, 32(1):23--40, 1986.

\bibitem{YanKum92}
Kyeongcheol Yang and P.~Vijay Kumar.
\newblock On the true minimum distance of {H}ermitian codes.
\newblock In {\em Coding theory and algebraic geometry (Luminy, 1991)}, volume
  1518 of {\em Lecture Notes in Math.}, pages 99--107. Springer, Berlin, 1992.

\end{thebibliography}

\def\lfhook#1{\setbox0=\hbox{#1}{\ooalign{\hidewidth
  \lower1.5ex\hbox{'}\hidewidth\crcr\unhbox0}}}





\end{document}